\documentclass[twocolumn,10pt]{IEEEtran}
\usepackage{bm,cite,float,amsmath,amssymb,amsthm}

\usepackage{algorithm}
\usepackage[noend]{algpseudocode}
\usepackage{indentfirst}
\usepackage{xcolor}

\makeatletter
\def\BState{\State\hskip-\ALG@thistlm}
\makeatother

\DeclareMathOperator\erf{erf}
\DeclareMathOperator\erfc{erfc}

\usepackage{amssymb}
\usepackage{amsmath}
\usepackage{graphicx}
\usepackage{cite}
\usepackage{citesort}
\usepackage{balance}
\usepackage[utf8]{inputenc}

\newtheorem{lemma}{Lemma}
\newtheorem{cor}{Corollary}
\newtheorem{ppro}{Proposition}

\IEEEoverridecommandlockouts

\usepackage{graphicx,epstopdf}
\usepackage{epsfig}	
\usepackage{amsfonts,balance}
\usepackage{bbm}
\floatname{algorithm}{Algorithm}
\usepackage{lipsum}

\newtheorem{theorem}{Theorem}

\newtheorem{corollary}{Corollary}

\makeatletter
\def\ScaleIfNeeded{%
\ifdim\Gin@nat@width>\linewidth \linewidth \else \Gin@nat@width
\fi } \makeatother

\begin{document}

\title{Modeling and Simulation of Molecular Communication Systems with a   Reversible Adsorption Receiver}

\author{Yansha~Deng,~\IEEEmembership{Member,~IEEE,}
        Adam~Noel,~\IEEEmembership{Member,~IEEE,}
         Maged~Elkashlan,~\IEEEmembership{Member,~IEEE,}\\
        Arumugam~Nallanathan,~\IEEEmembership{Senior Member,~IEEE,}
        and~Karen C. Cheung.

\vspace{-0.2cm}

\thanks{
Y. Deng and A. Nallanathan are with  Department of Informatics, King's College London, London, WC2R 2LS, UK (email:\{yansha.deng, arumugam.nallanathan\}@kcl.ac.uk).}
\thanks{A. Noel is with the School of Electrical Engineering and Computer Science, University of Ottawa, Ottawa, ON, K1N 6N5, Canada (email: anoel2@uottawa.ca).}
\thanks{ K. C. Cheung is with Department of Electrical and Computer Engineering,
University of British Columbia, Vancouver, BC,  V6T 1Z4, Canada (email:  kcheung@ece.ubc.ca).}
\thanks{M. Elkashlan is with  Queen Mary University of London, London E1 4NS,
UK (email: maged.elkashlan@qmul.ac.uk).}
}

\maketitle

\begin{abstract}
In this paper, we present an analytical model for  the diffusive molecular communication (MC) system with a reversible adsorption receiver in a fluid environment. The widely used concentration shift keying (CSK) is considered for modulation. The time-varying  spatial distribution of the  information molecules  under  the reversible adsorption and desorption reaction at the surface of a receiver is analytically characterized.  Based on the spatial distribution, we derive the net number of adsorbed information molecules expected in any time duration. We further derive the  net number of  adsorbed  molecules expected at the steady state to demonstrate the equilibrium concentration. Given the net number of adsorbed information molecules, the bit error probability of the proposed MC system  is analytically approximated. Importantly, we present a simulation framework for the proposed model that accounts for the diffusion and reversible reaction. Simulation results show the accuracy of our derived expressions, and demonstrate the positive effect of the adsorption rate and the negative effect of the desorption rate on the  error probability of reversible adsorption receiver with last transmit bit-1.  Moreover, our analytical results simplify to
the special cases of a full  adsorption receiver and a partial adsorption receiver, both of which do not include desorption. 
\end{abstract}

\begin{IEEEkeywords}
Molecular communication, reversible adsorption receiver, time varying spatial distribution, error probability.
\end{IEEEkeywords}

\maketitle

\section{Introduction}

 
Conveying information over a distance has been a problem for decades, and is  urgently demanded for multiple distance scales and various environments. The conventional solution is to utilize electrical- or electromagnetic-enabled  communication, which  is unfortunately inapplicable or inappropriate in very small dimensions or in specific  environments, such as in
salt water, tunnels, or human bodies.
Recent breakthroughs in bio-nano technology have motivated  molecular communication \cite{nakano2013molecu,nariman2016_survey} to be  a biologically-inspired technique for nanonetworks, where devices with functional components on the scale of 1--100 nanometers (i.e., nanomachines) share information over distance via chemical signals  in nanometer to micrometer scale environments.
These small scale bio-nanomachines are capable of encoding  information onto    physical molecules, sensing, and decoding the received information molecules,  which could enable  applications in drug delivery, pollution control,  health, and environmental monitoring \cite{nakano2012molecular}.


Based on the propagation channel, molecular communication (MC) can be classified into one of three categories: 1) Walkway-based MC, where molecules move directionally along  molecular rails  using  carrier substances, such as molecular motors \cite{moore2006design}; 2) Flow-based paradigm, where molecules propagate primarily via fluid flow. An example of this kind is the hormonal communication through the bloodstream in the human body\cite{nakano2013molecu}; 3) Diffusion-based MC, where molecules propagate via the random motion, namely Brownian motion,  caused by collisions with the fluid's molecules. In this case, molecule motion is less predictable, and the propagation is often assumed to follow the laws of a Wiener process. Examples  include  deoxyribonucleic acid (DNA) signaling among DNA segments \cite{howard1993radom630},  calcium signaling among cells \cite{nakano2005molecular}, and  pheromonal communication among animals \cite{bossert1963analysis}.

 Among the aforementioned three MC paradigms,  diffusion-based MC is the most simple,  general and energy efficient transportation  paradigm without the need for external energy or infrastructure. Thus,   research has  focused on the mathematical modeling and  theoretical analysis \cite{eckford2007nanoscale,kadloor2012molecular,kadloor2009framework,srinivas2012molecular,egan2014variance},  reception design \cite{noel2014improving}, receiver modeling \cite{yilmaz2014three844},  and  modulation and demodulation techniques \cite{yilmaz2014simulation,chun2015optimal,Yilmaz2014effect}, of  diffusion-based MC systems. 

In  diffusion-based MC, the transmit signal is encoded on the physical characteristics of information molecules (such as hormones, pheromones, DNA), which propagate through the fluid medium  via diffusion with the help of  thermal energy in the environment. The information can be encoded onto the the quantity,  identity, or released timing of the molecules. In the domain of timing channel, the first work on  diffusion based MC was pioneered by Eckford \cite{eckford2007nanoscale}, in which the propagation timing channel is ideally characterized as an additive noise channel.
In the domain of concentration-based encoding, the concentration level of information molecules represents different transmit signals.
Since the average displacement of an information molecule is directly proportional to the square root of diffusion time \cite{howard1993radom630},  long distance transmission requires much longer propagation times. Moreover, the randomness of the  arriving time for each molecule makes it difficult for the receiver to distinguish  between the signals transmitted in different bit intervals, because the number of received molecules in the current symbol  depends on the  molecules emitted in previous and current symbols. This is known as intersymbol interference (ISI).

 In most  existing literature, some  assumptions are made in order to  focus on the propagation channel. One  assumption is that each molecule is removed from the environment when it contributes once to the received signal. As such,  the information molecule concentration near the receiver is intentionally changed \cite{cuatrecasas1974membrane}. Another widely-used idealistic assumption is to consider a passive receiver, which is permeable to the information molecules passing by,  and is capable of counting the number molecules inside the receiver volume \cite{noel2014improving,noel2014optimal}.  The passive receiver model  easily encounters high ISI, since the same molecule may unavoidably contribute to the received signal many times in different symbol intervals. 
 
 
In a practical bio-inspired system, the surface of a receiver  is covered with  selective receptors, which  are sensitive to a specific type of information molecule (e.g., specific peptides or calcium ions). The surface of the receiver  may adsorb or  bind with this specific information molecule \cite{rospars2000perireceptor}. One example is  that the influx of calcium towards the  center of a receiver (e.g. cell) is induced by  the reception of a calcium signal \cite{bush2010nanoscale,Suda05exploratoryresearch}.
 
 Despite  growing research efforts in MC, the chemical reaction receiver has not been accurately characterized in most of the literature except by Yilmaz \cite{Yilmaz2014effect,yilmaz2014three844,yilmaz2014simulation} and Chou \cite{chou2013extended}. The primary challenge is accommodating the local reactions in the reaction-diffusion equations. In \cite{yilmaz2014three844} and \cite{schulten2000lectures}, the channel impulse response for   MC  with an absorbing receiver was derived. {The MolecUlar CommunicatIoN (MUCIN) simulator was presented in \cite{yilmaz2014simulation} to verify the fully-absorbing receiver. The results in \cite{yilmaz2014three844,yilmaz2014simulation} were then extended to   the ISI mitigation problem for the fully-absorbing receiver \cite{Yilmaz2014effect}.}
 In \cite{chou2013extended}, the mean and variance of the receiver output was derived  for MC  with a reversible reaction receiver  based on the reaction-diffusion master equation (RDME). The  analysis and simulations were performed using the subvolume-based method, where the transmitter and receiver were cubes, and the exact locations or placement of individual molecules were not captured. 
They considered the reversible reactions only happens inside the receiver (cube) rather than at the surface of receiver. 

Unlike   existing work on MC, we consider the \emph{reversible adsorption and desorption} ({A$\&$D}) receiver, which is capable of \emph{adsorbing} a certain type of information molecule near its surface, and desorbing the information molecules previously adsorbed at its surface. {A$\&$D}  is a widely-observed  process for   colloids \cite{FEDER1980144}, proteins \cite{ramsden1993concentration}, and polymers \cite{fang2005kinetics}. 
Within the  Internet of Bio-NanoThings (IoBNT), biological cells  are 
usually regarded as the substrates of the Bio-NanoThings.  These biological cells will be capable of interacting with each other by exchanging information, such as sensed chemical or physical parameters and sets of instructions or commands \cite{akyildiz2015bionano}. Analyzing the performance characteristics  of MC systems using biological cells   equipped with adsorption and desorption receptors allows for the comparison, classification, optimization and realization of different techniques to realize the IoBNT. 
The A$\&$D  process also simplifies to the special case of an \emph{absorbing} receiver (i.e., with no desorption). 
For consistency in this paper, we refer to receivers that do not desorb, but have infinite or finite absorption rates, as \emph{fully-adsorbing} and \emph{partially-adsorbing} receivers, respectively.

 From a theoretical perspective, researchers  have  derived the equilibrium concentration  of A$\&$D \cite{andrews2009accurate}, which is insufficient to model the time-varying channel impulse response (and ultimately the communications performance) of an {A$\&$D} receiver. 
Furthermore,  the simulation design for the  A$\&$D process of molecules at the surface of a \emph{planar} receiver   was also proposed   in \cite{andrews2009accurate}. However, the simulation procedure for a \emph{communication} model 
 with  a \emph{spherical} A$\&$D  receiver in a  fluid environment has never been solved and reported. In this model, information molecules are released by the transmission of pulses, propagate via free-diffusion through the channel, and contribute to the received signal via  A$\&$D at the receiver surface.
The challenges are the complexity in modeling the coupling effect of adsorption and desorption under diffusion, as well as accurately and dynamically tracking the location and the number of diffused molecules, adsorbed molecules and desorbed molecules (which are free to diffuse again).


Despite the aforementioned challenges, we  consider in this paper the  diffusion-based MC system with a point transmitter and an  {A$\&$D} receiver. The transmitter emits a certain number of information molecules at the start of  each symbol interval to represent the transmitted signal. These information molecules can adsorb to or desorb from the surface of the receiver. The number of information molecules adsorbed at the surface of the receiver is counted for information decoding. The  goal of this paper is to characterize the communications performance of an {A$\&$D}. Our major contributions are as follows:
\begin{enumerate}
\item We present an analytical model for the diffusion-based MC system with an {A$\&$D} receiver. We derive the exact expression for the channel impulse response at a spherical {A$\&$D} receiver in a three dimensional (3D) fluid environment  due to one instantaneous release of multiple molecules (i.e., single transmission).
\item  We derive the \emph{net}  number of adsorbed molecules expected at the surface of the {A$\&$D} receiver in any time duration. To measure the equilibrium concentration for a single transmission, we also derive the asymptotic  number of \emph{cumulative} adsorbed molecules expected at the surface of  {A$\&$D} receiver as time goes to infinity. 
\item  Unlike most literature  in \cite{noel2014optimal}, where the received signal is demodulated based on  the total number of molecules expected at the passive receiver,  we consider a simple demodulator based on the  net number of adsorbed molecules expected. When multiple bits are transmitted, the net number is more consistent than the total number.
\item  We apply the Skellam distribution to approximate  the net number of adsorbed molecules expected  at the surface of the {A$\&$D} receiver due to a single transmission of molecules. We formulate the bit error probability of the {A$\&$D} receiver using the Skellam distribution. Our results show the positive effect of adsorption rate and negative effect of desorption rate on the  error probability of {A$\&$D} receiver with last transmit bit-1.
\item We propose a simulation algorithm to simulate the diffusion, adsorption and desorption behavior of information molecules  based on a particle-based simulation framework.  Unlike  existing simulation platforms (e.g., Smoldyn \cite{andrews2010detailed}, N3sim \cite{llatser2011n3sim}), our simulation algorithm captures the  dynamic processes of  the MC system, which includes the   signal modulation, molecule free diffusion, molecule \text{A$\&$D} at the surface of the receiver, and signal demodulation. Our simulation results are in close agreement with the derived number of  adsorbed molecules expected. {Interestingly, we demonstrate that the error probability of the A$\&$D receiver for the last transmitted bit is worse at higher detection thresholds but better at low detection thresholds than both the full adsorption and partial adsorption receivers. This is because the A$\&$D receiver observes a lower peak number of adsorbed molecules but then a faster decay.}
\end{enumerate}

The rest of this paper is organized as follows. In Section II, we introduce the system model  with a single transmission at the transmitter and the {A$\&$D} receiver. In Section III, we present the channel impulse response of information molecules, i.e., the exact and asymptotic  number of adsorbed molecules
expected at the surface of the receiver. In Section IV, we derive the bit error probability of the proposed MC model due to multiple symbol intervals. In Section V, we  present the simulation framework. In Section VI, we discuss the numerical and simulation results. In Section VII, we conclude the contributions of this paper.

\section{System Model}
We consider a 3-dimensional (3D) diffusion-based MC system in a fluid environment with a point  transmitter and  a spherical A$\&$D  receiver. We assume  spherical symmetry where the transmitter is \emph{effectively} a spherical shell and the molecules are released from random points over the shell;  the actual angle to the transmitter when a molecule hits the receiver is ignored, so this assumption cannot accommodate a flowing environment.
 The point transmitter  is located at a distance $r_0$ from the center of the receiver and is at a distance $d=r_0-r_r$ from the nearest point on the surface of the receiver with radius $r_r$.   The extension to an asymmetric spherical model that accounts for the actual angle to the transmitter when a molecule hits the receiver complicates the derivation of the channel impulse response, and might be solved following \cite{scheider1972two}.  
 
 We assume all receptors are equivalent and can accommodate at most one adsorbed molecule.
 The ability of a molecule to adsorb at a given site is independent of the occupation of neighboring receptors.
The spherical receiver is assumed to have  no physical limitation on the number or placement of receptors on the receiver. Thus, there is no limit on the number of
molecules adsorbed to the receiver surface (i.e., we ignore saturation). This is an appropriate assumption for  a sufficiently low number of adsorbed molecules, or  for a sufficiently high concentration of receptors.

Once an information molecule binds to a
receptor site,  a physical response is activated to facilitate the counting of the molecule.
Generally,  due to  the non-covalent nature of  binding, in the dissociation process, the receptor may release
the adsorbed molecule to the 
fluid environment without changing its physical
characteristics, e.g., a ligand-binding receptor \cite{fan2005biosensors}.
  We also assume  perfect synchronization between the transmitter and the receiver as in most literature  \cite{kadloor2012molecular,kadloor2009framework,srinivas2012molecular,Yilmaz2014effect,chun2015optimal,yilmaz2014simulation,noel2014improving,yilmaz2014three844,noel2014optimal}. The system includes five processes: emission, propagation, reception, modulation and demodulation, which are detailed in the following.

\subsection{Emission}

The point transmitter releases one type of information molecule (e.g., hormones, pheromones) to the receiver for information transmission. The transmitter emits the information  molecules at $t=0$, where we define the initial condition as \cite[Eq.~(3.61)]{schulten2000lectures}
 \begin{align}
C\left( {r,\left. {t \to 0} \right|{r_0}} \right) = \frac{1}{{4\pi {r_0}^2}}\delta \left( {r - {r_0}} \right), \label{initial1}
\end{align}
where $C\left( {r,\left. {t \to 0} \right|{r_0}} \right)$  is the molecule distribution function at time ${t \to 0}$ and distance $r$ with initial distance $r_0$.

{We also define the first boundary condition as
 \begin{align}
\mathop {\lim }\limits_{r \to \infty } C\left( {r,\left. t \right|{r_0}} \right) = 0, \label{boundary3}
\end{align}
such that at arbitrary time,  the molecule distribution function   equals  zero when $r$ goes to infinity.}


\subsection{Diffusion}

 Once the information molecules are emitted, they diffuse by randomly colliding with other molecules in the environment. This random motion is called Brownian motion \cite{howard1993radom630}.
The concentration of information molecules is assumed to be sufficiently low that the collisions between those information molecules are ignored \cite{howard1993radom630}, such that each information molecule diffuses independently with constant diffusion coefficient  $D$.
 The propagation model in a 3D environment is described by  Fick's second law \cite{howard1993radom630,yilmaz2014three844}:
\begin{align}
\frac{{\partial \left( {r \cdot C\left( {r,\left. t \right|{r_0}} \right)} \right)}}{{\partial t}} = D\frac{{{\partial ^2}\left( {r \cdot C\left( {r,\left. t \right|{r_0}} \right)} \right)}}{{\partial {r^2}}}, \label{ficklaw}
\end{align}
where  the diffusion coefficient 
is usually obtained via experiment \cite{philip2008biological587}.

\subsection{Reception}
We consider a reversible A$\&$D receiver that is capable of counting the net number of adsorbed molecules at the surface of the receiver.  Any molecule that  hits the receiver surface is either adsorbed to the receiver surface  or reflected back into the fluid environment, based on the adsorption rate $k_1$ (length$\times $time$^{-1}$). The adsorbed molecules   either desorb  or remain stationary at the surface of receiver, based on the desorption rate $k_{-1}$ (time$^{-1}$). 

At $t=0$, there are no information molecules at the receiver surface, so the second initial condition is 
\begin{align}
C\left( {{r_r},\left. 0 \right|{r_0}} \right)=0, \text{and}\; {C_a}\left( {\left. 0 \right|{r_0}} \right) = 0, \label{initial2}
\end{align}
where ${C_a}\left( {\left. t \right|{r_0}} \right)$ is the average concentration of  molecules that are adsorbed to the receiver surface at time $t$. 

 For the solid-fluid interface located at $r_r$, the second boundary condition of the information molecules is  \cite[Eq. (4)]{andrews2009accurate}
\begin{align}
{\left. {D\frac{{\partial \left( {C\left( {r,\left. t \right|{r_0}} \right)} \right)}}{{\partial r}}} \right|_{r = r_r^ + }} = {k_1}C\left( {{r_r},\left. t \right|{r_0}} \right) - {k_{ - 1}}{C_a}\left( {\left. t \right|{r_0}} \right), \label{boundary1}
\end{align}	
which accounts for the adsorption and desorption reactions that can occur at the surface of the receiver.

Most generally, when both $k_1$ and $k_{-1}$ are  non-zero finite constants, \eqref{boundary1} is the boundary condition for the A$\&$D  receiver.
When $k_1 \to \infty$ and $k_{-1}=0$, \eqref{boundary1} is the boundary condition for the full adsorption (or fully-adsorbing) receiver, whereas when $k_1$ is a non-zero finite constant and $k_{-1}=0$, \eqref{boundary1} is the boundary condition for the partial adsorption (or partially-adsorbing) receiver. In these two special cases with $k_{-1}=0$, the lack of desorption results in more effective adsorption. Here, the adsorption rate $k_1$ is approximately limited to the thermal velocity of potential adsorbents (e.g., $k_1 < 7 \times 10^6\,\rm{\mu m /s}$ for a 50 kDa protein at 37 $^\circ$C) \cite{andrews2009accurate}; the desorption rate $k_{-1}$  is typically  between  $10^{-4}\,\rm{s^{-1}} $ and  $10^{4}\,\rm{s^{-1}}$ \cite{Tom2007Multistage}.

The surface concentration ${C_a}\left( {\left. t \right|{r_0}} \right)$ changes over time as follows:
\begin{align}
\frac{{\partial {C_a}\left( {\left. t \right|{r_0}} \right)}}{{\partial t}} = {\left. {D\frac{{\partial \left( {C\left( {r,\left. t \right|{r_0}} \right)} \right)}}{{\partial r}}} \right|_{r = r_r^ + }}, \label{boundary2}
\end{align}
which shows that the change  in the adsorbed concentration over time is equal to the flux of diffusion molecules towards the surface. 

Combining \eqref{boundary1} and  \eqref{boundary2}, we write
\begin{align}
\frac{{\partial {C_a}\left( {\left. t \right|{r_0}} \right)}}{{\partial t}} = {k_1}C\left( {{r_r},\left. t \right|{r_0}} \right) - {k_{ - 1}}{C_a}\left( {\left. t \right|{r_0}} \right), \label{boundary3}
\end{align}
which is known as the Robin or radiation boundary condition \cite{john1980math586,third1998gusta} and shows that the equivalent adsorption rate is proportional to the molecule concentration at the surface.

\subsection{Modulation and Demodulation}
In this model, we consider the widely applied amplitude-based modulation---concentration shift keying (CSK) \cite{yilmaz2014simulation,kuran2011modulation,Yilmaz2014effect,noel2014improving,kim2013novel}, where the concentration of information molecules is interpreted as the amplitude of the signal. Specifically, we utilize Binary CSK, where the transmitter emits $N_1$ molecules at the start of the bit interval
to represent the transmit bit-1, and emits $N_2$ molecules at the start of the bit interval to represent the transmit bit-0. To reduce the energy consumption and make the received signal more distinguishable, we assume  that $N_1=N_{\rm{tx}}$ and $N_2=0$.

We assume that the receiver is able to count the \emph{net} number of information molecules that are adsorbed to  the surface of the receiver in any sampling period by subtracting the number of
molecules bound to the surface of the receiver at the end of previous sampling time from
that at the end of current sampling time.
 The net number of adsorbed molecules  over a bit interval is then demodulated as the received signal for that bit interval. {This approach is in contrast to \cite{Yilmaz2014effect}, where the cumulative number of molecule arrivals in each symbol duration was demodulated as the received signal  (i.e., cumulative counter is reset to zero at each symbol duration).} We claim (and our results will demonstrate) that our approach is more appropriate for a simple demodulator.
 Here, we write the  net number of adsorbed molecules measured by the receiver in the $j$th  bit interval as ${N_{\rm{new}}^{\rm{Rx}}\left[ j \right] }$, and the decision threshold for the number of received molecules is $N_{\rm{th}}$. Using threshold-based demodulation, the receiver demodulates the received signal as bit-1 if ${N_{\rm{new}}^{\rm{Rx}}\left[ j \right] }  \ge N_{\rm{th}}$, and  demodulates the received signal as bit-0 if ${N_{\rm{new}}^{\rm{Rx}}\left[ j \right] } < N_{\rm{th}}$.

\section{Receiver Observations}
In this section, we first derive the spherically-symmetric spatial distribution  ${C\left( {r,\left. t \right|{r_0}} \right)}
$, which is the probability of finding a molecule at distance $r$ and time $t$. We then derive the  flux at  the surface of the A$\&$D receiver, from which we  derive the exact and asymptotic  number of adsorbed molecules expected at the surface of the receiver.

\subsection{Exact Results}
The time-varying spatial distribution of information molecules at the surface of the receiver is an important statistic for capturing the molecule concentration in the diffusion-based MC system. We solve it in the following theorem.

\begin{theorem}
The expected time-varying spatial distribution of an information molecule released into a 3D fluid environment with a reversible adsorbing receiver is given by
\begin{align}
C\left( {r,\left. t \right|{r_0}} \right) = & \; \frac{1}{{8\pi {r_0}r\sqrt {\pi Dt} }}\exp \left\{ { - \frac{{{{\left( {r - {r_0}} \right)}^2}}}{{4Dt}}} \right\}
\nonumber\\&  + \frac{1}{{8\pi {r_0}r\sqrt {\pi Dt} }}\exp \left\{ { - \frac{{{{\left( {r + {r_0} - 2{r_r}} \right)}^2}}}{{4Dt}}} \right\}
\nonumber\\&  - \frac{1}{{2\pi r}}\int_0^\infty  {\left( {{e^{ - jwt}}\varphi _Z^*\left( w \right) + {e^{jwt}}{\varphi _Z}\left( w \right)} \right){\mathrm{d}}w},
\label{der20}
\end{align}
where 
 \begin{align}
{\varphi _Z}\left( w \right) = & \;Z\left( {jw} \right) = \frac{{2\left( {\frac{1}{{{r_r}}} + \frac{{{k_1}jw}}{{D\left( {jw + {k_{-1}}} \right)}}} \right)}}{{\left( {\frac{1}{{{r_r}}} + \frac{{{k_1}jw}}{{D\left( {jw + {k_{-1}}} \right)}} + \sqrt {\frac{{jw}}{D}} } \right)}}
\nonumber\\& \times \frac{1}{{8\pi {r_0}\sqrt {Djw} }}\exp \left\{ { - \left( {r + {r_0} - 2{r_r}} \right)\sqrt {\frac{{jw}}{D}} } \right\},
\label{der19}
\end{align}
and $\varphi _Z^*\left( w \right)$ is the complex conjugate of $\varphi _Z\left( w \right)$.
\end{theorem}
\begin{proof}
See Appendix A.
\end{proof}
Our results in \textbf{Theorem 1} can be easily computed using Matlab. 
 We observe that \eqref{der20} reduces to an \emph{absorbing} receiver \cite[Eq.~(3.99)]{schulten2000lectures} when there is no desorption (i.e., $k_{-1}=0$).

To characterize the number of information molecules adsorbed to the surface of the receiver using ${ C\left( {\left. {r,t} \right|{r_0}} \right)}$,
we define the  rate of the coupled reaction (i.e.,  adsorption and desorption)  at the surface of the A$\&$D receiver   as \cite[Eq.~(3.106)]{schulten2000lectures}
\begin{align}
K\left( {\left. t \right|{r_0}} \right) = 4\pi r_r^2D{\left. {\frac{{\partial C\left( {\left. {r,t} \right|{r_0}} \right)}}{{\partial r}}} \right|_{r = {r_r}}}.
\label{der21}
\end{align}

\begin{cor}
The rate of the coupling reaction at  the surface of a reversible adsorbing receiver is given by
\begin{align}
K\left( {\left. t \right|{r_0}} \right) = & \; 2{r_r}D{\int_{ 0 }^\infty  {{e^{ - jwt}}\left[ {\sqrt {\frac{{jw}}{D}} {\varphi _Z}\left( w \right)} \right]} ^*}dw
\nonumber\\& + 2{r_r}D{\int_{ 0 }^\infty  {{e^{  jwt}}\left[ {\sqrt {\frac{{jw}}{D}} {\varphi _Z}\left( w \right)} \right]} }dw,
\label{der22}
\end{align}
where ${\varphi _Z}\left( w \right)$ is as given in \eqref{der19}. 
\end{cor}
\begin{proof}
By substituting \eqref{der20} into \eqref{der21}, we derive the coupling reaction  rate at the surface of an A$\&$D receiver as \eqref{der22}.
\end{proof}

From \textbf{Corollary 1}, we can derive the  net change in the number of adsorbed molecules expected for any time interval in the following theorem.

\begin{theorem}
With a single emission at $t=0$, the  net change in the number of  adsorbed molecules expected  at the surface of the A$\&$D receiver during the interval [$T$, $T$+$T_{s} $]  is derived as 
\begin{align}
& \mathbb{E}\left[ {{N_{\rm{A\&D}}}\left( {\left. {{\Omega _{{r_r}}},T,T + T_{s}} \right|{r_0}} \right)} \right] = 2{r_r}{N_{\rm{tx}}}D \nonumber\\
  &     \hspace{0.3cm} \times \Bigg[ {\int_{ 0 }^\infty  \frac{{{e^{ - jwT}} - {e^{ - jw\left( {T + T_{s}} \right)}}}}{{jw}}\Big[ {\sqrt {\frac{{jw}}{D}} {\varphi _Z}\left( w \right)} \Big] ^*}dw \Bigg.
  \nonumber\\
  &     \hspace{0.3cm} \Bigg.+  {\int_{ 0 }^\infty  {\frac{{{e^{  jw\left( {T + T_{s}} \right)}}-{e^{  jwT}} }}{{jw}}\Big[ {\sqrt {\frac{{jw}}{D}} {\varphi _Z}\left( w \right)} \Big]} }dw \Bigg]
  ,
\label{der241}
\end{align}
where ${\varphi _Z}\left( w \right)$ is given in \eqref{der19},  $T_{s}$ is the sampling time, and  ${\Omega _{{r_r}}}$ represents the spherical receiver with radius $r_r$.
\end{theorem}

\begin{proof}
The cumulative fraction of particles that are adsorbed to the receiver surface at time $T$  is expressed as
\begin{align}
{{R_{\rm{A\&D}}}\left( {\left. {{\Omega _{{r_r}}},T} \right|{r_0}} \right)}&  = \int_0^T {K\left( {\left. t \right|{r_0}} \right)dt} 
\nonumber\\& \hspace{-2cm} = 2{r_r}D \Bigg[{\int_{0 }^\infty  {\frac{{1 - {e^{ - jwT}}}}{{jw}}\left[ {\sqrt {\frac{{jw}}{D}} {\varphi _Z}\left( w \right)} \right]} ^*}dw \Bigg.
\nonumber\\& \hspace{-1.6cm} \Bigg.+{\int_{ 0 }^\infty  {{\frac{{{e^{jwT}} - 1}}{{jw}}}\left[ {\sqrt {\frac{{jw}}{D}} {\varphi _Z}\left( w \right)} \right]} }dw \Bigg].
\label{der23}
\end{align}

Based on \eqref{der23}, the  net change in adsorbed molecules expected at the receiver surface during the interval [$T$, $T$+$T_{s} $] is defined as
\begin{align}
\mathbb{E}&\left[ {{N_{{\rm{A\&D}}}}\left( {\left. {{\Omega _{{r_r}}},T,T + T_{s}} \right|{r_0}} \right)} \right] = 
\nonumber\\& {N_{\rm{tx}}}{R_{\rm{A\&D}}}\left( {\left. {{\Omega _{{r_r}}},T + T_{s}} \right|{r_0}} \right)   - {N_{\rm{tx}}}{R_{\rm{A\&D}}}\left( {\left. {{\Omega _{{r_r}}},T} \right|{r_0}} \right).
\label{der24}
\end{align}

Substituting \eqref{der23} into \eqref{der24}, we derive the expected net change of adsorbed molecules during any observation interval as \eqref{der241}.
\end{proof}

Note that the net change in the number of adsorbed molecules in each \emph{bit interval}  will be recorded at the receiver, which will be converted to  the recorded net change of adsorbed molecules in each \emph{bit interval}, and   compared with the decision threshold $N_{\rm{th}}$ to demodulate the received signal (the sampling interval is smaller than one bit interval).

\subsection{Asymptotic Behavior: Equilibrium Concentration}

In this section, we are interested in the asymptotic number of adsorbed molecules due to a single emission as  $T_b$ goes to infinity, i.e., the concentration of adsorbed molecules at the steady state.  Note that this asymptotic concentration of adsorbed molecules  is an important quantity that influences the  number of adsorbed molecules expected in  subsequent bit intervals, and we have assumed that the receiver surface has infinite receptors.
 Thus, in the remainder of this section,
we  derive the cumulative number of adsorbed molecules expected at the surface of the  A$\&$D receiver, the partial adsorption  receiver, and the  full adsorption  receiver,
 as ${T_b} \to \infty$.

\subsubsection{Reversible A$\&$D  Receiver}
\begin{lemma}
As $T_b \to \infty $, the cumulative number of adsorbed molecules  expected at the A$\&$D  receiver  simplifies to
\begin{align}
\mathbb{E} &\left[ {N_{\rm{A\&D}}}\left( {\left. {{\Omega _{{r_r}}},T_b \to \infty} \right|{r_0}} \right) \right] = \frac{{N_{\rm{tx}}}{{r_r}}}{{2 {r_0}}}
 \nonumber\\&  - 4{N_{\rm{tx}}}{r_r}D\int_0^\infty  {\frac{1}{w}{\mathop{\rm Im}\nolimits} \left[ {\sqrt {\frac{{jw}}{D}} {\varphi _Z}\left( w \right)} \right]} dw.
\label{der25}
\end{align}

\end{lemma}
\begin{proof}
We express  the cumulative fraction of particles  adsorbed to the surface of the A$\&$D   receiver at time $T_b$ in \eqref{der23} as
\begin{align}
&{{R_{\rm{A\&D}}}\left( {\left. {{\Omega _{{r_r}}},{T_b}} \right|{r_0}} \right)} \nonumber\\& = {\mathop{\rm Re}\nolimits} \left[ {4{r_r}D\int_0^\infty  {\frac{{{e^{jw{T_b}}} - 1}}{{jw}}\left( {\sqrt {\frac{{jw}}{D}} {\varphi _Z}\left( w \right)} \right)} dw} \right]
\nonumber\\& =  4{r_r}D\int_0^\infty  {\frac{{\sin w{T_b}}}{w}{\mathop{\rm Re}\nolimits} \left[ {\sqrt {\frac{{jw}}{D}} {\varphi _Z}\left( w \right)} \right]} dw
\nonumber\\& \hspace{0.4cm} + 4{r_r}D\int_0^\infty  {\frac{{\cos w{T_b} - 1}}{w}{\mathop{\rm Im}\nolimits} \left[ {\sqrt {\frac{{jw}}{D}} {\varphi _Z}\left( w \right)} \right]} dw
\nonumber\\&  = 4{r_r}D\int_0^\infty  {\frac{{\sin z}}{z}{\mathop{\rm Re}\nolimits} \left[ {q\left( {\frac{z}{{T_b}}} \right)} \right]} dz  + 4{r_r}D\int_0^\infty  {\frac{{\cos z}}{z}} 
\nonumber\\& \hspace{0.5cm}
{\mathop{\rm Im}\nolimits} \left[ {q\left( {\frac{z}{{T_b}}} \right)} \right]dz  - 4{r_r}D\int_0^\infty  {\frac{1}{w}{\mathop{\rm Im}\nolimits} \left[ {q\left( w \right)} \right]} dw,
\label{proofA1}
\end{align}
where
\begin{align}
q\left( w \right) = &\frac{{\left( {\frac{1}{{{r_r}}} + \frac{{{k_1}jw}}{{D\left( {jw + {k_{-1}}} \right)}}} \right)}}{{\left( {\frac{1}{{{r_r}}} + \frac{{{k_1}jw}}{{D\left( {jw + {k_{-1}}} \right)}} + \sqrt {\frac{{jw}}{D}} } \right)}}\frac{1}{{4\pi {r_0}D}}
\nonumber\\& \times \exp \left\{ { - \left( {{r_0} - {r_r}} \right)\sqrt {\frac{{jw}}{D}} } \right\}.
\end{align} 


As ${T_b} \to \infty $, we have the following:
\begin{align}
&\mathbb{E} \left[ {N_{\rm{A\&D}}}\left( {\left. {{\Omega _{{r_r}}},{T_b} \to \infty} \right|{r_0}} \right) \right] =  4{r_r}D{N_{\rm{tx}}}\Big[ \int_0^\infty  {\frac{{\sin z}}{z}}{\mathop{\rm Re}\nolimits} \left[ {q\left( 0 \right)} \right] \Big.
 \nonumber\\&  \hspace{0.5cm} \Big. dz + \int_0^\infty  {\frac{{\cos z}}{z}{\mathop{\rm Im}\nolimits} \left[ {q\left( 0 \right)} \right]} dz
  - \int_0^\infty  {\frac{1}{w}{\mathop{\rm Im}\nolimits} \left[ {q\left( w \right)} \right]} dw \Big]
  \nonumber\\& 
  \mathop  = \limits^{(b)} 4{r_r}D{N_{\rm{tx}}}\Big[ \int_0^\infty  {\frac{{\sin z}}{z}{\mathop{\rm Re}\nolimits} \left[ {q\left( 0 \right)} \right]} dz - \int_0^\infty  {\frac{1}{w}{\mathop{\rm Im}\nolimits} \left[ {q\left( w \right)} \right]} dw\Big]
   \nonumber\\& 
   \mathop  = \limits^{(c)} {N_{\rm{tx}}}\Big[\frac{{{r_r}}}{{\pi {r_0}}}\int_0^\infty  {\frac{{\sin z}}{z}} dz - 4{r_r}D\int_0^\infty  {\frac{1}{w}{\mathop{\rm Im}\nolimits} \left[ {q\left( w \right)} \right]} dw\Big]
  \nonumber\\&  
   = \frac{{N_{\rm{tx}}}{{r_r}}}{{2 {r_0}}}
  - 4{N_{\rm{tx}}}{r_r}D\int_0^\infty  {\frac{1}{w}{\mathop{\rm Im}\nolimits} \left[ {\sqrt {\frac{{jw}}{D}} {\varphi _Z}\left( w \right)} \right]} dw,
\label{proofA2}
\end{align}
where $(b)$ is due to the fact that ${\mathop{\rm Im}\nolimits} \left[ {q\left( 0 \right)} \right]=0$,  and $(c)$ is due to $q\left( 0 \right) = \frac{1}{{4\pi {r_0}D}}$.
\end{proof}

\subsubsection{Partial Adsorption Receiver}
The partial adsorption receiver  only adsorbs some of the molecules that collide with its surface, corresponding to $k_1$ as a finite constant and $k_{-1} =0$ in \eqref{boundary1}. 

\begin{ppro}
 The  number of molecules expected to be adsorbed to the partial adsorption receiver by time ${T_b}$, as ${T_b} \to \infty $,  is derived as
\begin{align}
\mathbb{E} &\left[ {N_{\mathrm{PA}}}\left( {\left. {{\Omega _{{r_r}}},{T_b} \to \infty} \right|{r_0}} \right) \right] = \frac{{{N_{\mathrm{tx}}}{k_1}r_r^2}}{{{r_0}\left( {{k_1}{r_r} + D} \right)}}.
\label{der28}
\end{align}
\end{ppro}

\begin{proof} 
We note that the exact expression for the  net number of adsorbed molecules expected at the partial adsorption receiver during [${T}$, ${T}$+$T_{s} $]  can be  derived from \cite[Eq.~(3.114)]{schulten2000lectures} as
\begin{align}
\mathbb{E}&\left[ {{N_{\rm{PA}}}\left( {\left. {{\Omega _{{r_r}}},{T},{T} + T_{s}} \right|{r_0}} \right)} \right] = N_{\rm{tx}} \frac{{{r_r}\alpha  - 1}}{{{r_0}\alpha }}
\nonumber\\& \times \Bigg[ {\erf\left\{ {\frac{{{r_r} - {r_0}}}{{\sqrt {4D\left( {{T} + {T_s}} \right)} }}} \right\}} \Bigg. - \exp \left\{ {\left( {{r_0} - {r_r}} \right)\alpha } \right.
\nonumber\\&\left. { + D\left( {{T} + {T_s}} \right){\alpha ^2}} \right\}\erfc\left\{ {\frac{{{r_0} - {r_r} + 2D\alpha \left( {{T} + {T_s}} \right)}}{{\sqrt {4D\left( {{T} + {T_s}} \right)} }}} \right\} 
\nonumber\\& - \erf\left\{ {\frac{{{r_r} - {r_0}}}{{\sqrt {4D{T}} }}} \right\} + \exp \left\{ {\left( {{r_0} - {r_r}} \right)\alpha  + D{T}{\alpha ^2}} \right\}
\nonumber\\& \Bigg. \times {\erfc\left\{ {\frac{{{r_0} - {r_r} + 2D\alpha {T}}}{{\sqrt {4D{T}} }}} \right\}} \Bigg],
\label{der24pab}
\end{align}
where $\alpha  = \frac{{{k_1}}}{D} + \frac{1}{{{r_r}}}$.
 
The cumulative fraction of molecules adsorbed at the partial adsorption receiver by time ${T_b}$ was derived in \cite[Eq.~(3.114)]{schulten2000lectures}  as
\begin{align}
&{R_{\rm{PA}}}\left( {\left. {{\Omega _r},{T_b}} \right|{r_0}} \right) =\frac{{{r_r}\alpha  - 1}}{{{r_0}\alpha }}\left( {1 + \erf\left\{ {\frac{{{r_r} - {r_0}}}{{\sqrt {4D{T_b}} }}} \right\}} \right.
\nonumber\\& \hspace{0.1cm} \left. { - \exp \left\{ {\left( {{r_0} - {r_r}} \right)\alpha  + D{T_b}{\alpha ^2}} \right\}\erfc\left\{ {\frac{{{r_0} - {r_r} + 2D\alpha {T_b}}}{{\sqrt {4D{T_b}} }}} \right\}} \right).
\label{der29}
\end{align}

By setting ${T_b} \to \infty $ and taking the expectation of \eqref{der29}, we arrive at \eqref{der28}.
\end{proof}

The asymptotic result in \eqref{der28} for the  partial adsorption receiver reveals that the  number of adsorbed molecules expected  at infinite time ${T_b}$ increases with increasing  adsorption rate $k_1$, and decreases with increasing  diffusion coefficient $D$  and increasing distance between the transmitter and the center of the receiver $r_0$.

\subsubsection{Full Adsorption Receiver}
In the full adsorption receiver, all molecules adsorb when they collide with its surface, which corresponds to the case of $k_1 \to \infty$ and $k_{-1} = 0$ in \eqref{boundary1}. 

\begin{ppro}
 The  cumulative number of adsorbed molecules expected at the full adsorption receiver by time $T_b$, as ${T_b} \to \infty $,  is derived as
\begin{align}
\mathbb{E} &\left[ {N_{\mathrm{FA}}}\left( {\left. {{\Omega _{{r_r}}},{T_b} \to \infty} \right|{r_0}} \right) \right] = \frac{{N_{{\rm{tx}}}}{{r_r}}}{{{r_0}}}.
\label{der26}
\end{align}
\end{ppro}

\begin{proof}  
We note that the exact expression for the  net number of adsorbed molecules expected at the full adsorption receiver during [$T$, $T$+$T_{s} $] has been derived in \cite{schulten2000lectures,yilmaz2014three844}  as
\begin{align}
\mathbb{E}&\left[ {{N_{{\mathrm{FA}}}}\left( {\left. {{\Omega _{{r_r}}},T,T + T_{s}} \right|{r_0}} \right)} \right] = 
\nonumber\\& {N_{{\rm{tx}}}}\frac{{{r_r}}}{{{r_0}}}\Bigg[ {{\rm{erfc}}\left\{ {\frac{{{r_0} - {r_r}}}{{\sqrt {4D\left( {T + {T_s}} \right)} }}} \right\} - {\rm{ erfc}}\left\{{ \frac{{{r_0} - {r_r}}}{{\sqrt {4DT} }}}\right\} } \Bigg].
\label{der24fab}
\end{align} 

The fraction of molecules adsorbed to the full adsorption receiver by time ${T_b}$ was derived in \cite[Eq.~(3.116)]{schulten2000lectures} and \cite[Eq.~(32)]{yilmaz2014three844} as
\begin{align}
{R_{\mathrm{FA}}}\left( {\left. {{\Omega _r},{T_b}} \right|{r_0}} \right) = \frac{{{r_r}}}{{{r_0}}}\erfc\left\{ {\frac{{{r_0} - {r_r}}}{{\sqrt {4D{T_b}} }}} \right\}.
\label{der27}
\end{align}

By setting $T_b \to \infty $ and taking the expectation of \eqref{der27}, we arrive at \eqref{der26}.
\end{proof}

Alternatively, with the help of integration by parts,
the result in \eqref{der25}  reduces to the asymptotic result in \eqref{der26} for the full adsorption receiver   by setting $k_{1}=\infty$ and $k_{-1}=0$ .

The asymptotic result  for the full adsorption receiver in \eqref{der26} reveals that the  cumulative number of adsorbed molecules  expected by infinite time ${T_b}$ is
\emph{independent} of the diffusion coefficient, and directly proportional to the ratio between the radius of receiver and  the distance between the transmitter and the center of receiver.

\section{Error Probability} 
In this section, we propose that the  net number of adsorbed molecules in a bit interval be used for receiver demodulation. We also derive the error probability of the MC system using the Poisson approximation and the Skellam distribution.

To calculate the error probability at the receiver, we first need to model the statistics of molecule adsorption. 
For a single emission at $t=0$, the net number of molecules adsorbed during $[T,T+T_b]$ is \emph{approximately} modeled as the difference between two binomial distributions as
\begin{align}
N_{{\rm{new}}}^{{\rm{Rx}}}\sim & B\left( {{N_{{\rm{tx}}}},R_{\rm{A\&D}}\left( {\left. {{\Omega _{{{\rm{r}}_{\rm{r}}}}},{{T}} + {{{T}}_b}} \right|{{\rm{r}}_{\rm{0}}}} \right)} \right) - 
\nonumber \\& \hspace{2.5cm} B\left( {{N_{{\rm{tx}}}},R_{\rm{A\&D}}\left( {\left. {{\Omega _{{{\rm{r}}_{\rm{r}}}}},{{T}}} \right|{{\rm{r}}_{\rm{0}}}} \right)} \right), \label{approsx}
\end{align}
where the cumulative fraction of  particles that are adsorbed
to  the A$\&$D  receiver ${R}_{\rm{A\&D}}\left( {\left. {{\Omega _{{r_r}}},{T}} \right|{r_0}} \right)$ is   given in   \eqref{proofA1}.  Note that the number of molecules adsorbed at $T+T_b$ depends on that at $T$, however  this dependence can be ignored for a sufficiently large bit interval, and makes \eqref{approsx}  accurate. 
The  number of adsorbed molecules represented by Binomial distribution can also be approximated  using either the  Poisson distributions or the Normal distributions. 

The net number of adsorbed molecules depends on the emission in the current bit interval and those in previous bit intervals.  Unlike the full adsorption receiver  in \cite{yilmaz2014arrival,Yilmaz2014effect,heren2014effect,yilmaz2014arrival} and partial adsorption receiver where the net number of adsorbed molecules is always positive, the net number of adsorbed molecules of the A$\&$D receiver \emph{can be negative}. Thus, we cannot model the net number of adsorbed molecules of the reversible adsorption receiver during one bit interval as 
$N_{{\rm{new}}}^{{\rm{Rx}}}\sim B\left( {{N_{{\rm{tx}}}},{R}\left( {\left. {{\Omega _{{{\rm{r}}_{\rm{r}}}}},{{T,T}} + {{{T}}_b}} \right|{{\rm{r}}_{\rm{0}}}} \right)} \right)$ with  ${R}\left( {\left. {{\Omega _{{{\rm{r}}_{\rm{r}}}}},{{T,T}} + {{{T}}_b}} \right|{{\rm{r}}_{\rm{0}}}} \right){\rm{ = }}\int_{{T}}^{{{T}} + {{{T}}_b}} {{\rm{K}}\left( {\left. {\rm{t}} \right|{{\rm{r}}_{\rm{0}}}} \right){\rm{dt}}}$, which was  used to   model that of full adsorption receiver and partial adsorption receiver \cite{heren2014effect,yilmaz2014arrival}.

For multiple emissions, the cumulative number of adsorbed molecules is modeled as the sum of  multiple binomial random variables. This sum does not
 lend itself to a convenient expression. Approximations for the sum were used in \cite{kuran2010energy}.
Here,  the binomial distribution can be approximated with the Poisson distribution, when we have sufficiently large $N_{\rm{tx}}$ and sufficiently small ${R_{\rm{A\&D}}}\left( {\left. {{\Omega _{{r_r}}},  {T}} \right|{r_0}} \right)$ \cite{natrella2010nist}. Thus, we approximate the net number of  adsorbed molecules received in the $j$th bit interval as 
\begin{align}
{N_{\rm{new}}^{\rm{Rx}}\left[ j \right] } \sim & P\left( {\sum\limits_{i = 1}^j {{N_{{\rm{tx}}}}{s_i}R_{\rm{A\&D}}\left( {\left. {{\Omega _{{r_r}}},\left( {j - i + 1} \right){T_b}} \right|{r_0}} \right)} } \right)
\nonumber\\& \hspace{0.1cm} - P\left( {\sum\limits_{i = 1}^{j } {{N_{{\rm{tx}}}}{s_i}R_{\rm{A\&D}}\left( {\left. {{\Omega _{{r_r}}},\left( {j - i} \right){T_b}} \right|{r_0}} \right)} } \right), \label{poisson}
\end{align}
where $s_i$ is the   $i$th transmitted bit. The difference between two Poisson random variables follows the Skellam distribution \cite{karlis2003analysis}. For  threshold-based demodulation, the error probability of the transmit bit-1 signal in the $j$th 
bit is then
\begin{align}
&{P_e}\left[ {{{\hat s}_j} = 0\left| {{s_j} = 1,{s_{1:j - 1}}} \right.} \right]  
\nonumber\\& \hspace{0.1cm} = \Pr \left( {\left. {N_{{\rm{new}}}^{{\rm{Rx}}}\left[ j \right] < {N_{{\rm{th}}}}} \right|{s_j} = 1,{s_{1:j - 1}}} \right)
\nonumber\\&\hspace{0.1cm} {\approx  \sum\limits_{n =  - \infty }^{{N_{{\rm{th}}}-1}} {\exp \left\{ { - \left( {{\Psi _1} + {\Psi
 _2}} \right)} \right\}} {\left( {{{{\Psi
 _1}} \mathord{\left/
 {\vphantom {{{\Psi
 _1}} {{\Psi
 _2}}}} \right.
 \kern-\nulldelimiterspace} {{\Psi
 _2}}}} \right)^{{n \mathord{\left/
 {\vphantom {n 2}} \right.
 \kern-\nulldelimiterspace} 2}}} {I_n}\left( {2\sqrt {{\Psi
 _1}{\Psi
 _2}} } \right),}
 \label{error_bit1}
\end{align}
where 
\begin{align}
{\Psi
 _1} = \sum\limits_{i = 1}^j {{N_{{\rm{tx}}}}{s_i}R_{\rm{A\&D}}\left( {\left. {{\Omega
 _{{r_r}}},\left( {j - i + 1} \right){T_b}} \right|{r_0}} \right)}, \label{omega_1} 
\end{align}
\begin{align}
{\Psi
 _2} = \sum\limits_{i = 1}^{j - 1} {{N_{{\rm{tx}}}}{s_i}R_{\rm{A\&D}}\left( {\left. {{\Omega _{{r_r}}},\left( {j - i} \right){T_b}} \right|{r_0}} \right)}, \label{omega_2} 
\end{align}
${{\hat s}_j}$ is the detected $j$th bit, and ${I_n}\left(  \cdot  \right)$ is the modified Bessel function of the first kind.

 Similarly, the error probability of the transmit bit-0 signal  in the $j$th 
bit is given as 
\begin{align}
&{P_e}\left[ {{{\hat s}_j} = 1\left| {{s_j} = 0,{s_{1:j - 1}}} \right.} \right] 
\nonumber\\& \hspace{0.1cm} = \Pr \left( {\left. {N_{{\rm{new}}}^{{\rm{Rx}}}\left[ j \right] \ge {N_{{\rm{th}}}}} \right|{s_j} = 0,{s_{1:j - 1}}} \right)
\nonumber\\&\hspace{0.1cm}  { \approx \sum\limits_{n = {N_{{\rm{th}}}}}^\infty  {\exp \left\{ { - \left( {{\Psi
 _1} + {\Psi
 _2}} \right)} \right\}} {\left( {{{{\Psi
 _1}} \mathord{\left/
 {\vphantom {{{\Psi
 _1}} {{\Psi
 _2}}}} \right.
 \kern-\nulldelimiterspace} {{\Psi
 _2}}}} \right)^{{n \mathord{\left/
 {\vphantom {n 2}} \right.
 \kern-\nulldelimiterspace} 2}}}{I_n}\left( {2\sqrt {{\Psi
 _1}{\Psi
 _2}} } \right),}
 \label{error_bit0}
\end{align}
where $\Psi
 _1$ and $\Psi
 _1$ are given in \eqref{omega_1} and \eqref{omega_2}, respectively.

{Thus, the  error probability of the random transmit bit in the $j$th interval is  expressed by
\begin{align}
{P_e}\left[ j \right] = &{P_1}{P_e}\left[ {{{\hat s}_j} = 0\left| {{s_j} = 1,{s_{1:j - 1}}} \right.} \right]
\nonumber\\&  + {P_0}{P_e}\left[ {{{\hat s}_j} = 1\left| {{s_j} = 0,{s_{1:j - 1}}} \right.} \right], \label{overallerror}
\end{align}
where $P_1$ and $P_0$ denotes the probability of sending bit-1 and bit-0, respectively.

For  comparison, we also present the error probability of the transmit bit-1 signal  in the $j$th 
bit and error probability of the transmit bit-0 signal  in the $j$th 
bit for the full adsorption receiver and the partial adsorption receiver using the Poisson approximation as  
 \begin{align}
&{P_e}\left[ {{{\hat s}_j} = 0\left| {{s_j} = 1,{s_{1:j - 1}}} \right.} \right]  
{\approx  \exp \left\{ {{N_{{\rm{tx}}}}\Gamma } \right\}\sum\limits_{n = 0}^{{N_{{\rm{th}}}}-1} {\frac{{{{\left[ {{N_{{\rm{tx}}}}\Gamma } \right]}^n}}}{{n!}}} ,}
 \label{error_bitFA1}
\end{align}
and
\begin{align}
&{P_e}\left[ {{{\hat s}_j} = 1\left| {{s_j} = 0,{s_{1:j - 1}}} \right.} \right]  
{\approx 1- \exp \left\{ {{N_{{\rm{tx}}}}\Gamma } \right\}\sum\limits_{n = 0}^{{N_{{\rm{th}}}}-1} {\frac{{{{\left[ {{N_{{\rm{tx}}}}\Gamma } \right]}^n}}}{{n!}}} .}
 \label{error_bitFA0}
\end{align}
In \eqref{error_bitFA1} and \eqref{error_bitFA0}, we have 
 \begin{align}
 \Gamma  = \sum\limits_{i = 1}^j {{s_i}{R_{\rm{FA}}}\left( {\left. {{\Omega _{{{{r}}_{{r}}}}},\left( {{{j - i}}} \right){{{T}}_{{b}}}\left( {{{j - i + 1}}} \right){{{T}}_{{b}}}} \right|{{{r}}_{{0}}}} \right)} 
 \end{align} 
 for the full adsorption receiver, and 
 \begin{align}
 \Gamma  = \sum\limits_{i = 1}^j {{s_i}{R_{\rm{PA}}}\left( {\left. {{\Omega _{{{{r}}_{{r}}}}},\left( {{{j - i}}} \right){{{T}}_{{b}}}\left( {{{j - i + 1}}} \right){{{T}}_{{b}}}} \right|{{{r}}_{{0}}}} \right)} 
 \end{align}  for the partial adsorption receiver.}

\section{Simulation Framework}
This section describes the stochastic simulation framework for the point-to-point MC system with the A$\&$D  receiver described by \eqref{boundary1}, which can be simplified to the MC system with the partial adsorption receiver and full adsorption receiver by setting $k_{-1}=0$ and $k_1=\infty$, respectively. This simulation framework takes into account the signal modulation, molecule free diffusion, molecule \text{A$\&$D} at the surface of the receiver, and signal demodulation.  

To  model the stochastic reaction of molecules in the fluid, two options are a subvolume-based simulation framework  or a particle-based simulation framework.  In a subvolume-based simulation framework, the environment is divided into many subvolumes, where the number of molecules in each subvolume is recorded \cite{chou2013extended}.  In a particle-based simulation framework \cite{steve2004stochastic},  the exact position of each molecule and the  number of molecules in the fluid environment is recorded. 
To   accurately capture the locations of individual information molecules, we adopt a particle-based simulation framework with a spatial resolution on the order of several nanometers \cite{steve2004stochastic}.

\subsection{Algorithm}
We  present the algorithm for simulating the MC system with an  A$\&$D  receiver in Algorithm 1. 
In the following subsections, we describe the details of Algorithm 1.
\begin{algorithm}
\caption{The Simulation of a MC System with an A$\&$D  Receiver}\label{euclid}
Require: $N_{\rm{tx}}$, $r_0$, $r_r$,  $ {\Omega _{{r_r}}}$, $D$, $\Delta t$, $T_s$, $T_b$, $N_{\rm{th}}$
\begin{algorithmic}[1]
\Procedure{Initialization}{}
\State Generate Random Bit Sequence  $\left\{ {{b_1},{b_2}, \cdots ,{b_j}, \cdots } \right\}$ 
\State Determine Simulation End Time
\BState  \textbf{For all} Simulation Time Step \textbf{do} 
\BState  \, \textbf{If} at start of $j$th bit interval and $b_j= ``1"$
\State  Add $N_{\text{tx}}$ emitted molecules  
\BState \,  \textbf{For all} free molecules in environment \textbf{do} 
\State  Propagate  free molecules following  $\mathcal{N}\left( {0,2D\Delta t} \right)$
\State  Evaluate  distance $d_m$ of  molecule to receiver  
   \If {$d_m< r_r$}
\State     Update  state $\&$ location of  collided molecule 
\State   Update $\#$ of collided  molecules  $N_C$
\EndIf
\BState \,   \textbf{For all} $N_C$ collided molecules  \textbf{do}
 \If {Adsorption Occurs}
\State  Update $\#$ of \emph{newly}-adsorbed molecules  $N_A$
\State  Calculate adsorbed molecule  location 
\State \, \, $\left( {x_m^A,y_m^A,z_m^A} \right)$
\Else  
\State Reflect the molecule off receiver surface to
\State \, \, $\left( {x_m^{Bo},y_m^{Bo},z_m^{Bo}} \right)$
\EndIf
\BState  \, \textbf{For all} \emph{previously}-adsorbed molecules \textbf{do} 
\If {Desorption Occurs}
\State    Update  state $\&$ location of  desorbed molecule 
\State  Update $\#$  of \emph{newly}-desorbed molecules  $N_D$
\State  Displace \emph{newly}-desorbed molecule to 
\State $\left( {x_m^D,y_m^D,z_m^D} \right)$
\EndIf
\BState \, Calculate \emph{net}  number of  adsorbed molecules, 
\BState \, which is   $N_A-N_D$ 
\BState Add net number of adsorbed molecules in each simulation interval of $j$th bit interval to determine ${N_{\rm{new}}^{\rm{Rx}}\left[ j \right] }$
\BState Demodulate by comparing ${N_{\rm{new}}^{\rm{Rx}}\left[ j \right] }$ with $N_{\rm{th}}$
\EndProcedure
\end{algorithmic}
\end{algorithm}

\subsection{Modulation, Emission, and Diffusion}
In our model, we consider  BCSK, where two different numbers of molecules represent the binary signals ``1" and ``0". At the start of each bit interval, if the current bit is ``1",  then  $N_{\rm{tx}}$ molecules are emitted from the point transmitter at a distance $r_0$ from the center of the receiver. Otherwise, the point transmitter emits no molecules to transmit bit-0.

 The time is divided into small simulation intervals of size $\Delta t$, and each time instant is $t_m = m\Delta t$, where $m$ is the current simulation index. According to   Brownian motion,  the displacement of a molecule in each dimension in one simulation step $\Delta t$ can be modeled by an independent Gaussian distribution with variance $2D\Delta t$ and zero mean ${\mathcal{N}\left( {0,2D\Delta t} \right)}$ . The displacement  $\Delta S$ of a molecule in a 3D fluid environment in one simulation step $\Delta t$ is therefore 
\begin{align}
 \Delta S = \left\{ {\mathcal{N}\left({0,2D\Delta t} \right),\; \mathcal{N}\left({0,2D\Delta t} \right),\;\mathcal{N}\left({0,2D\Delta t} \right)} \right\}.
 \end{align}

In each simulation step, the number of molecules and their locations are stored.
\subsection{Adsorption or Reflection}
 According to the second boundary condition in \eqref{boundary2},  molecules that collide with the receiver surface are either adsorbed or reflected back. The $N_C$ collided  molecules    are identified  by  calculating the distance between each molecule and the center of the receiver. Among the collided molecules, the probability of a molecule being adsorbed to the receiver surface, i.e., the adsorption probability, is a function of the diffusion coefficient, which is given as
 \cite[Eq.~(10)]{erban2007reactive}
\begin{align}
{P_A} = {k_1}\sqrt {\frac{{\pi \Delta t}}{D}} . \label{abpro}
 \end{align}
 
The probability  that a collided molecule bounces off of the receiver is $1-{P_A}$.


 It is known that  adsorption may occur during the simulation step $\Delta t$, and  determining exactly where a molecule  adsorbed to the surface of the receiver during $\Delta t$ is a non-trivial problem. Unlike \cite{andrews2009accurate}  (which considered a flat adsorbing surface), we assume that the molecule's adsorption site during $[t_{m-1}, 
 t_m]$ is the location where the line, formed by this molecule's location at the start of the current simulation step $\left( {{x_{m - 1}},{y_{m - 1}},{z_{m - 1}}} \right)$ and this molecule's location  at the end of the current simulation step after diffusion $\left( {{x_{m}},{y_{m}},{z_{m}}} \right)$, intersects the surface of the receiver. Assuming that the location of the center of receiver is $(x_r,y_r,z_r)$, then the location of the intersection point between this 3D line segment,
 and a sphere with center at  $(x_r,y_r,z_r)$ in the $m$th simulation step, can be shown to be
{
 \begin{align}
x_m^A = & {x_{m - 1}} + \frac{{{x_m} - {x_{m - 1}}}}{{\Lambda }}g,\label{loc1}
\\ y_m^A = & {y_{m - 1}} + \frac{{{y_m} - {y_{m - 1}}}}{{\Lambda  }}g, \label{loc11}
\end{align}
\begin{align}
z_m^A = & {z_{m - 1}} + \frac{{{z_m} - {z_{m - 1}}}}{{\Lambda }}g, \label{loc12}
 \end{align}
where 
\begin{align}
\Lambda = \sqrt {{{\left( {{x_m} - {x_{m - 1}}} \right)}^2} + {{\left( {{y_m} - {y_{m - 1}}} \right)}^2} + {{\left( {{z_m} - {z_{m - 1}}} \right)}^2}} , \label{loc_delta}
\end{align}
\begin{align}
g = \frac{{ - b - \sqrt {{b^2} - 4ac} }}{{2a}}. \label{loc2}
\end{align}}

{
In \eqref{loc2}, we have 
\begin{align}
a = &{\left( {\frac{{{x_m} - {x_{m - 1}}}}{{\Lambda}}} \right)^2} + {\left( {\frac{{{y_m} - {y_{m - 1}}}}{{\Lambda }}} \right)^2} + {\left( {\frac{{{z_m} - {z_{m - 1}}}}{{\Lambda }}} \right)^2},\nonumber\\
b = &2\frac{{\left( {{x_m} - {x_{m - 1}}} \right) {({x_{m - 1}}-x_r)} }}{{\Lambda }}  + 2\frac{{\left( {{y_m} - {y_{m - 1}}} \right)}{{(y_{m - 1}-y_r)} }}{{\Lambda }}
\nonumber\\&    + 2\frac{{\left( {{z_m} - {z_{m - 1}}} \right) {({z_{m - 1}}-z_r) } }}{{\Lambda }},\label{loc31}
\end{align}
\begin{align}
c = & { {({x_{m - 1}}-x_r)}^2} + { {({y_{m - 1}}-y_r)}^2} + {{({z_{m - 1}}-z_r)}^2}-{r_r}^2,
 \label{loc3}
\end{align}
where $\Lambda$ is given in \eqref{loc_delta}.}

 Of course, due to symmetry, the location of the adsorption site does not impact the overall accuracy of the  simulation.

If a molecule fails to adsorb to the receiver, then in the reflection process we make the approximation that the molecule bounces back to its position at the start of the current  simulation step.
 Thus, the location of the molecule after reflection by the receiver in the $m$th simulation step is approximated as
\begin{align}
\left( {x_m^{Bo},y_m^{Bo},z_m^{Bo}} \right) = \left( {{x_{m - 1}},{y_{m - 1}},{z_{m - 1}}} \right). \label{loc4}
\end{align}  

Note that the approximations for molecule locations in the adsorption process and the reflection process can be  accurate for sufficiently small simulation steps (e.g., $\Delta t < 10^{-7}$  s  for the system that we simulate in Section V), but  small simulation steps result in poor computational efficiency.  The tradeoff between the accuracy and the efficiency can be deliberately balanced by the choice of simulation step.

\subsection{Desorption}
In the desorption process, the molecules adsorbed at the receiver boundary  either desorb or remain adsorbed.  The desorption process can be modeled as a  first-order chemical reaction. 
 Thus,  the desorption probability of a molecule at the receiver surface  during  $\Delta t$ is given by  \cite[Eq. (22)]{andrews2009accurate}
\begin{align}
{P_D} = 1 - {e^{ - {k_{ - 1}}\Delta t}}.\label{loc5}
\end{align} 

The displacement of a  molecule after desorption is an important factor for accurate modeling of  molecule behaviour.  If the simulation step were small, then we might place the desorbed molecule near the receiver surface; otherwise, doing so may result in an artificially higher chance of re-adsorption in the following time step, resulting in an inexact  concentration profile. To avoid this, we take into account the diffusion \emph{after} desorption, and place the desorbed molecule away from the surface with displacement $\left( {\Delta x,\Delta y,\Delta z} \right)$ 
\begin{align}
\left( {\Delta x,\Delta y,\Delta z} \right) = \left( {f\left( {{P_1}} \right),f\left( {{P_2}} \right),f\left( {{P_3}} \right)} \right), \label{loc6}
\end{align} 
where  each component was empirically found to be \cite[Eq.~(27)]{andrews2009accurate}
\begin{align}
f\left( P \right) = \sqrt {2D\Delta t} \frac{{0.571825P - 0.552246{P^2}}}{{1 - 1.53908P + 0.546424{P^2}}}. \label{loc7}
\end{align} 

In \eqref{loc6}, $P_1$, $P_2$ and $P_3$ are  uniform random numbers between 0 and 1. { Placing the desorbed molecule at a random distance away  from where the molecule was adsorbed  may not be sufficiently accurate due to the lack of consideration for the coupling effect of  A$\&$D  and the diffusion coefficient in \eqref{loc7}.}  

Unlike \cite{andrews2009accurate}, we have a spherical receiver, such that a molecule
after desorption in our model must be  displaced differently.
We assume that the location of a molecule after desorption $\left( {x_m^D,y_m^D,z_m^D} \right)$, based on its location at the start of the current simulation step and the location of the center of the receiver $(x_r,y_r,z_r)$, can be approximated as
{
\begin{align}
x_m^D = & {x_{m - 1}^A} + {\rm{sgn}}\left( {{x_{m - 1}^A} - {x_r}} \right)\Delta x,
\nonumber\\ y_m^D = & {y_{m - 1}^A} + {\rm{sgn}}\left( {{y_{m - 1}^A} - {y_r}} \right)\Delta y,
\nonumber\\ z_m^D = & {z_{m - 1}^A} + {\rm{sgn}}\left( {{z_{m - 1}^A} - {z_r}} \right) \Delta z.
\label{loc8}
\end{align}}
In \eqref{loc8}, $\Delta x$, $\Delta y$, and $\Delta z$ are given in \eqref{loc6}, and $\rm{sgn}\left(  \cdot  \right)$ is the Sign function.

\subsection{Demodulation}
The receiver is capable of counting the net change in the number of adsorbed molecules in each bit interval. The   net number of adsorbed molecules for an entire bit interval is compared with the threshold $N_{\rm{th}}$ and demodulated as the received signal.

\section{Numerical Results}
In this section, we examine the channel response and the asymptotic channel response due to a single bit transmission. We also examine  the channel response  and the error probability  due to multiple bit transmissions. In all figures of  this section, we use FA, PA, ``Anal.'' and
``Sim." to abbreviate ``Full adsorption receiver'', ``Partial adsorption receiver'', ``Analytical'' and ``Simulation'', respectively.  Also, the units for the adsorption rate  $k_1$ and desorption rate $k_{-1}$ are $\rm{\mu m/s}$ and $\rm{s^{-1}}$ in all figures, respectively. 
{{In Figs.~\ref{fig:1} to \ref{fig:4}, we set the parameters according to  micro-scale cell-to-cell communication\footnote{{The small separation distance between the transmitter and receiver compared to the receiver radius follows from the example of the  pancreatic islets, where the average  cell  size is around 15 micrometers and the communication range is around 1$- $15 micrometers  \cite{yilmaz2014three844,yilmaz2014simulation}.}},\footnote{This diffusion coefficient value corresponds to that of a large molecule, however, our  analytical results and simulation algorithm apply to any specific value.}: 
 $N_{\rm{tx}} = 1000$, $r_r = 10$ $\rm{\mu m}$, $r_0 = 11 \; \rm{\mu m} $, $D = 8$ $\rm{\mu m^2/s}$,  and the sampling interval $T_s =0.002$ s.}}

\subsection{Channel Response   }

 \begin{figure}[t!]
    \begin{center}
        \includegraphics[width=3.0 in,height=2.5in]{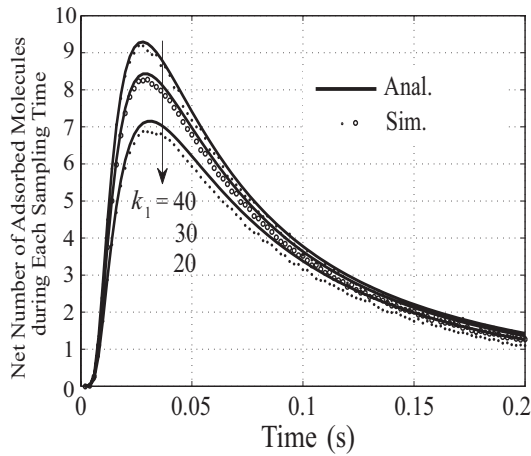}
        \caption{The  net number of adsorbed  molecules for various adsorption rates with 
        $k_{-1} = 5$ $\rm{s^{-1}} $ and   the simulation step $\Delta t =10^{-5}$ s.}
        \label{fig:1}
    \end{center}
\end{figure}

Figs.~\ref{fig:1} and \ref{fig:2} plot the  net change of  adsorbed molecules   at the surface of the A$\&$D  receiver during each sampling time $T_s$ due to a single bit transmission.
The expected analytical curves are plotted using the  exact result in \eqref{der241}. The simulation points are plotted by measuring the net change of  adsorbed molecules  during $ [t,t + T_s]$   using  Algorithm 1 described in  Section IV, where $t = nT_s$, and $n \in \{1,2,3,\ldots \} $.   In both figures, we  average  the net number of adsorbed molecules expected over 10000 independent emissions of $N_{\rm{tx}}=1000$ information molecules at time $t=0$.  We see that the expected net number of adsorbed molecules measured using simulation is close to the exact analytical curves. The small gap between the curves results from the local approximations in the adsorption, reflection, and desorption processes in \eqref{abpro}, \eqref{loc4}, and \eqref{loc8}, which can be reduced by setting a smaller simulation step.

Fig.~\ref{fig:1} examines the impact of the adsorption rate on the net number of adsorbed  molecules expected at the surface of the receiver.  We fix the desorption rate to be $k_{-1} = 5\,\rm{s^{-1}}$. The expected net number of adsorbed  molecules increases with increasing  adsorption rate $k_1$,  as predicted by \eqref{boundary1}.  
 \mbox{Fig.~\ref{fig:2}}  shows the impact of the desorption rate on the expected net number of adsorbed molecules at the surface of the receiver. We set $k_1 = 20\,\mu \rm{m /s}$. The net number of adsorbed  molecules expected  decreases with increasing  desorption rate $k_{-1}$, which is as predicted by \eqref{boundary1}.
{From a communication perspective,  
Fig.~\ref{fig:1}   shows that 
 a   higher adsorption rate makes the bit-1 signal more distinguishable, whereas Fig.~\ref{fig:2} shows that a lower desorption rate makes the bit-1 signal more distinguishable for the decoding process.}
In Figs.~\ref{fig:1} and \ref{fig:2},  the shorter tail due to the lower adsorption rate and the higher desorption rate corresponds to less intersymbol interference.  
 \begin{figure}[t!]
    \begin{center}
        \includegraphics[width=3.0 in,height=2.5in]{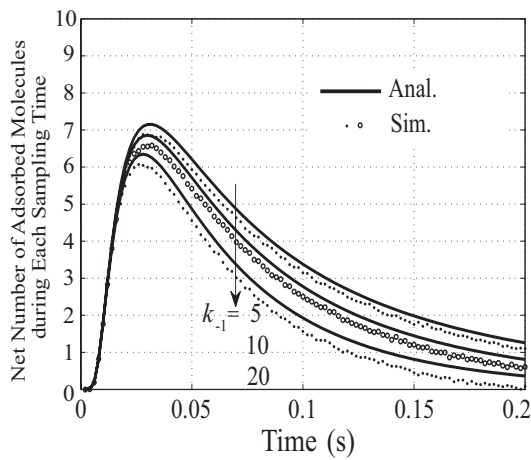}
        \caption{The  net number of adsorbed  molecules for various desorption rates with
        $k_1 = 20$ $\rm{\mu m /s}$ and   the simulation step $\Delta t =10^{-5}$ s. }
        \label{fig:2}
    \end{center}
\end{figure}


 \begin{figure}[t!]
    \begin{center}
        \includegraphics[width=3.0 in,height=2.5in]{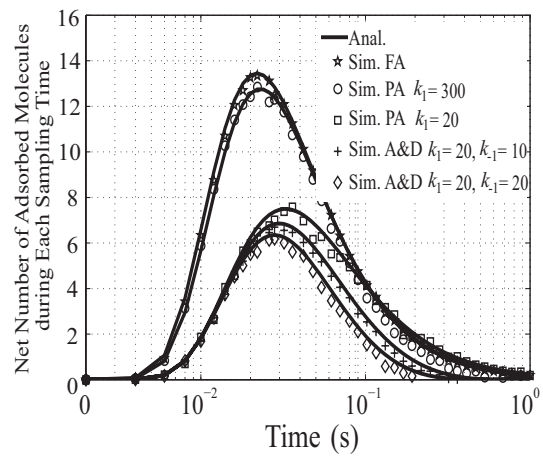}
        \caption{The net number of adsorbed  molecules  with the simulation step $\Delta t =10^{-5}$ s.}
        \label{fig:3}
    \end{center}
\end{figure}

Fig.~\ref{fig:3} plots the net number of adsorbed molecules from 1 bit transmission over a longer time scale. We compare the A$\&$D  receiver with  other receiver designs in order to compare their intersymbol interference (ISI). The analytical curves for the A$\&$D  receiver, the partial adsorption receiver, and the full adsorption receiver are plotted using the  expressions in \eqref{der241}, \eqref{der24pab}, and \eqref{der24fab}, respectively.  
The markers are plotted by measuring the net number of adsorbed molecules during $[t,t+T_s]$ for one bit interval using Algorithm 1 described in  Section IV. We see a close match between the analytical curves and the simulation curves, which confirms the correctness of our derived results. 

{It is clear from Fig.~\ref{fig:3} that the full adsorption receiver and the partial adsorption receiver with high adsorption rate have  longer ``tails". 
Interestingly, the A$\&$D  receiver in our model has the shorter tail, even though it has the same adsorption rate $k_1$ as one of the partial adsorption receivers. This might be surprising since the A$\&$D  receiver would have more total adsorption events than the partial adsorption receiver with the same $k_1$. 
The reason for this difference is that the desorption behaviour at the surface of the receiver results in more adsorption events, but not more  \emph{net} adsorbed molecules; molecules that desorb are not counted unless they adsorb again.

As expected, we see the highest peak $\mathbb{E}\left[ {{N}\left( {\left. {{\Omega _{{r_r}}},T,T + T_{s}} \right|{r_0}} \right)} \right]$ in Fig.~\ref{fig:3} for the full adsorption receiver, which is because all molecules colliding with the surface of the receiver are adsorbed.
 For the partial adsorption receiver, the peak value of $\mathbb{E}\left[ {{N}\left( {\left. {{\Omega _{{r_r}}},T,T + T_{s}} \right|{r_0}} \right)} \right]$ increases with increasing  adsorption rate $k_1$ as shown in \eqref{boundary1}.  The net number of   adsorbed molecules expected  at the  partial adsorption receiver is higher than that at the  A$\&$D  receiver with the same $k_1$. This means  the full adsorption receiver and the partial adsorption receiver have more distinguishable received signals between  bit-1  and bit-0, compared with the A$\&$D  receiver.

\subsection{Equilibrium Concentration }
 \textcolor{red}{
\begin{figure}[t!]
    \begin{center}
        \includegraphics[width=3.0 in,height=2.5in]{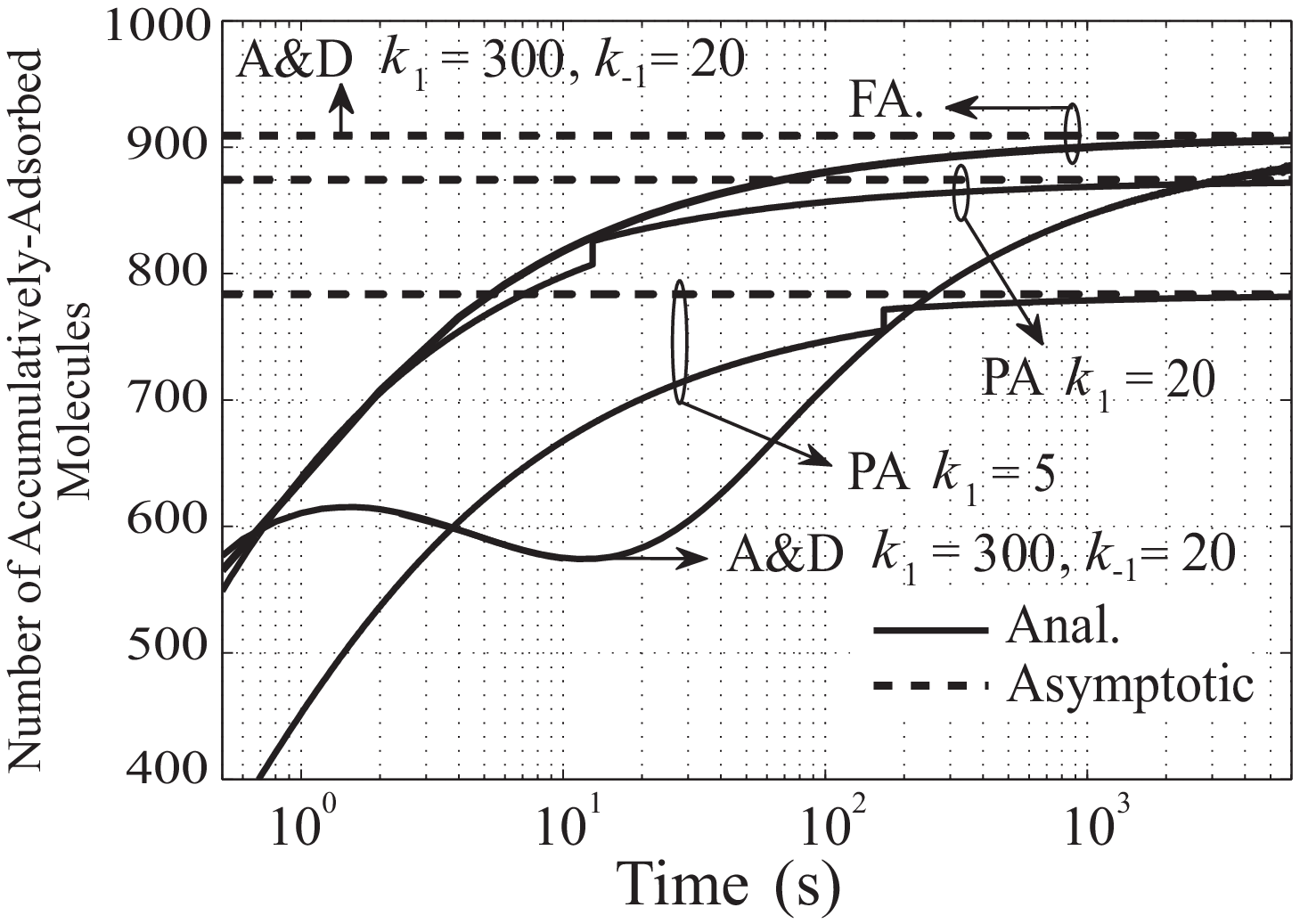}
        \caption{The cumulative number of  adsorbed molecules. }
        \label{fig:4}
    \end{center}
\end{figure} }

Fig.~\ref{fig:4} plots the  number of \emph{cumulatively}-adsorbed   molecules  expected at the surface of the different types of receiver with a single emission $N_{\rm{tx}}$ and  as $T_b \to \infty $. The solid curves are plotted by accumulating the  net number of adsorbed molecules expected in each sampling time $\mathbb{E}\left[ {{N}\left( {\left. {{\Omega _{{r_r}}},T, T+T_s } \right|{r_0}} \right)} \right]$ in \eqref{der24}, \eqref{der24pab}, and \eqref{der24fab}. The dashed lines are plotted  using the derived asymptotic expressions in \eqref{der25}, \eqref{der28}, and \eqref{der26}. The asymptotic analytical lines  are
in precise agreement with the exact analytical curves as $T_b \to \infty$. The exact analytical curves of the full adsorption receiver and the partial adsorption receiver converge to their own asymptotic analytical lines faster than the convergence of the A$\&$D  receiver. Interestingly, we find that the analytical curve of the A$\&$D   receiver decreases after increasing over a few bit intervals, and then increases again, while that of the partial adsorption receiver has an increasing trend as time goes large and shows a sudden jump at a specific time. {The discontinuities  in the PA curves are caused by the underflow during the evaluation of $\text{erfc}(x)$, which results from the limitation of  Matlab's smallest possible double.}
 As expected, the asymptotic curve of the partial adsorption receiver degrades with decreasing $k_1$, as shown in \eqref{der28}.
More importantly, the full adsorption  receiver has a higher initial accumulation rate but the \emph{same} asymptotic number of bound molecules as that of the A$\&$D  receiver with $k_1 = 300$ $\mu \rm{m /s}$ and $k_{-1}=20$ $\rm{s^{-1}} $.

\subsection{Demodulation Criterion}

\begin{figure}[t!]
    \begin{center}
        \includegraphics[width=3.0 in,height=2.5in]{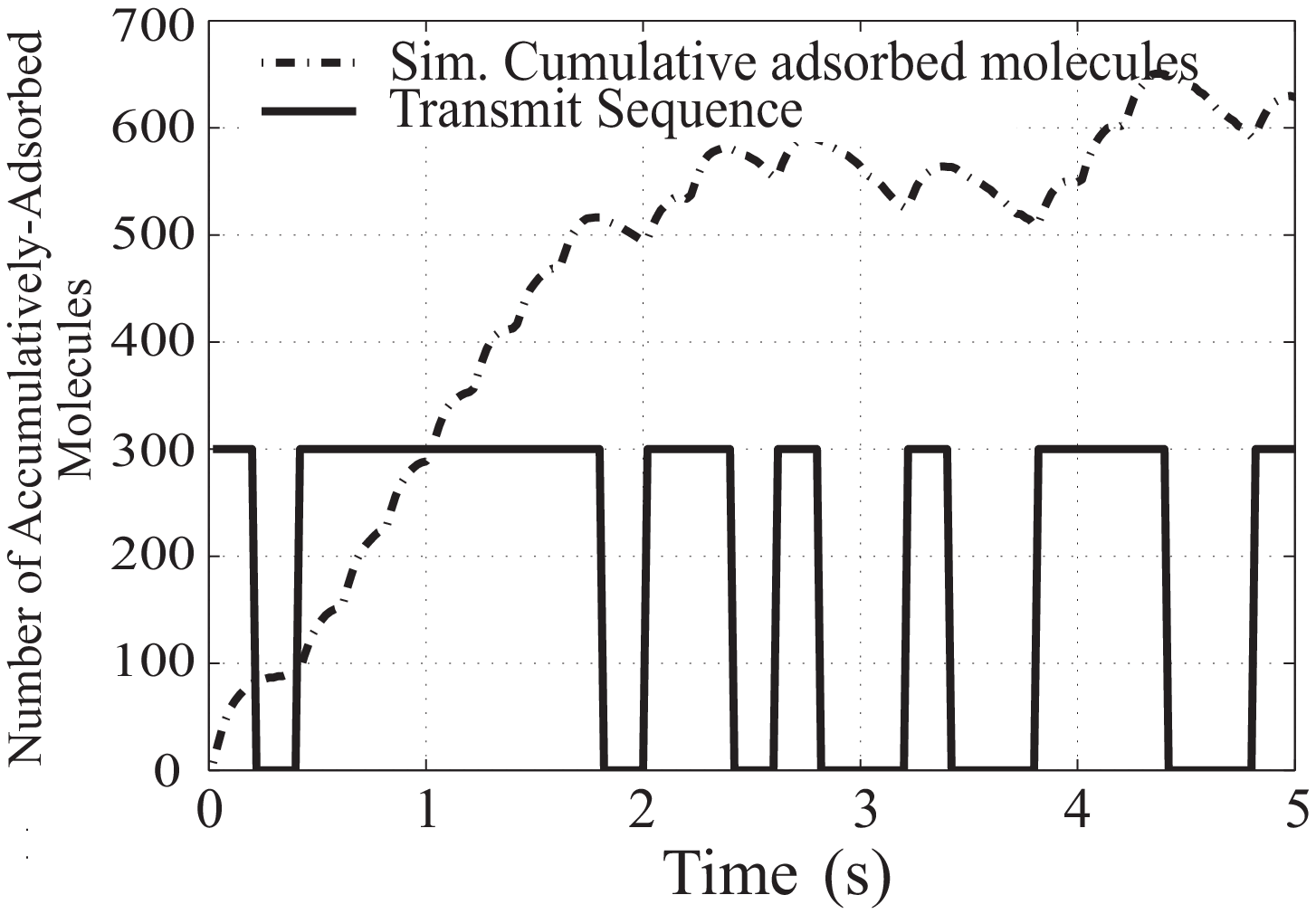}
        \caption{The cumulative number of  adsorbed molecules. }
        \label{fig:61}
    \end{center}
\end{figure} 
In Figs.~\ref{fig:61} and \ref{fig:62}, we compare our proposed demodulation criterion using the \emph{net} number of adsorbed molecules  with the widely used demodulation criterion using the  number of \emph{cumulatively}-adsorbed  molecules in \cite{noel2014improving,noel2014optimal}. In these two figures, we set the parameters:  $k_1=10$ $\rm{\mu m /s}$, $k_{-1}=5$ $\rm{s^{-1}} $,
         $N_{\rm{tx}} = 300$, $r_r = 10$ $\rm{\mu m}$, $r_0 = 11$ $\rm{\mu m }$, $D = 8$ $\rm{\mu m^2/s} $, $\Delta t =10^{-5}$ $s$,  $T_s =0.02$ s,  the bit interval $T_b =0.2$ s, and the number of bits $N_b =25$.
Fig.~\ref{fig:61}  plots the  number of cumulatively-adsorbed  molecules expected at the surface of the A$\&$D  receiver in each sampling time due to the  transmission of multiple bits, whereas Fig.~\ref{fig:62}  plots the net number of adsorbed molecules expected at the surface of the A$\&$D  receiver at each sampling time due to the  transmission of multiple bits. In both figures, the solid lines plot the transmit sequence, where each bit can be bit-0 or bit-1. Note that in both figures,  the y-axis values of the transmit signal for bit-0 are zero, and those for bit-1  are scaled in order to clearly show the relationship between the transmit sequence and the number of adsorbed molecules.   
The dashed lines are plotted by  averaging  the  number of  adsorbed molecules over 1000 independent emissions  for the same generated transmit  sequence in the simulation.
\begin{figure}[t!]
    \begin{center}
        \includegraphics[width=3.0 in,height=2.5in]{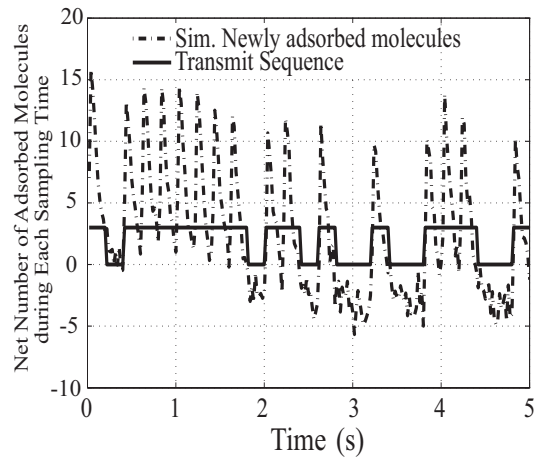}
        \caption{The net number of adsorbed molecules. }
        \label{fig:62}
    \end{center}
\end{figure} 

In Fig.~\ref{fig:61}, it is shown that the  number of \emph{cumulatively}-adsorbed molecules expected at the surface of the A$\&$D  receiver increases in  bit-1 bit intervals, but can decrease in bit-0 bit intervals.  This is because  the new information molecules injected  into the environment due to bit-1 increases  the number of cumulatively-adsorbed molecules, whereas, without new molecules due to bit-0, the desorption reaction can eventually decrease the cumulative number of adsorbed molecules. 
In Fig.~\ref{fig:62}, we observe a single peak net number of adsorbed molecules for each bit-1 transmitted, similar to the channel response for a single bit-1 transmission in Fig.~\ref{fig:1}.  We also see a noisier signal in each bit-0 interval due to the ISI effect brought by the previous transmit signals.


To motivate our proposed demodulation criterion, we compare the behaviours of the accumulatively and net change of adsorbed molecules at the receiver in Fig.~\ref{fig:61} and Fig.~\ref{fig:62}. We see that the  number of cumulatively-adsorbed molecules increases with increasing time, whereas
the met number of  adsorbed molecules have comparable value (between 10 and 15) for all bit-1 signals. As such, the threshold for demodulating the  number of cumulatively-adsorbed molecules should be increased as time increases, while the same threshold can be used to demodulate the net number of adsorbed molecules in different bit intervals. We claim that   the received signal should be demodulated using  the net number of  adsorbed molecules. Note that the net number of adsorbed molecules refers to the \emph{net} change, since the receiver cannot  distinguish between the molecules that just adsorbed and those that were already adsorbed.


\begin{figure}[t!]
    \begin{center}
        \includegraphics[width=3.0 in,height=2.5in]{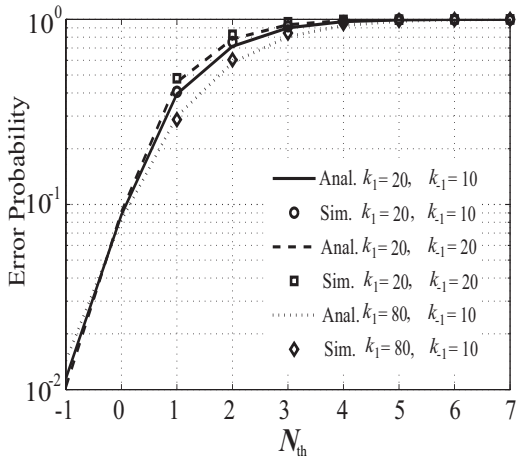}
        \caption{The error probability for the last transmit bit-1. }
        \label{fig:7}
    \end{center}
\end{figure} 

\subsection{Error Probability}

Figs.~\ref{fig:7} and \ref{fig:8} plot the error probability as a function of decision threshold for the third bit in a 3-bit sequence where the last bit is bit-1 and bit-0, respectively. The first 2 bits are ``1 1". {In these two figures, we set the parameters: $N_{\rm{tx}} = 50$, $r_r = 15$ $\rm{\mu m}$, $r_0 = 20$ $\rm{\mu m} $, $D = 5$ $\rm{\mu m^2/s} $, $\Delta t =10^{-6}$ s,  $T_s =0.002$ s,  and the bit interval $T_b =0.2$ s.} Note that with lower diffusion coefficient and larger distance between the transmitter and the receiver, a weaker signal is observed.
The  simulation results are compared with the evaluation of \eqref{error_bit1} for bit-1 and \eqref{error_bit0} for bit-0, where the  net number of adsorbed molecules expected at the surface of the receiver are approximated by the Skellam distribution. 
There are negative thresholds with meaningful error probabilities, thus confirming the need for the Skellam distribution.
{ The simulation points are plotted by averaging the total errors  over $10^5$ independent emissions  of  transmit  sequences with last bit-1 and bit-0. In both figures, we see a close match between the
simulation points and the analytical lines. }

Fig.~\ref{fig:7} plots the error probability of the last transmit bit-1 at the A$\&$D  receiver with $N_b = 3$ bits transmitted for various adsorption rate $k_1$ and desorption rate $k_{-1}$. We see that the error probability  of the last transmit bit-1 increases monotonically with increasing   threshold $N_{\rm{th}}$.
Interestingly, we find that for the same $k_{-1}$, the error probability improves with increasing $k_1$. This can be explained by the fact that increasing $k_1$ increases the amplitude of the  net number of adsorbed molecules expected (as shown in Fig.~\ref{fig:1}), which makes the received signal for bit-1 more distinguishable than that for  bit-0.  For the same $k_1$, the error probability degrades with increasing $k_{-1}$, which is because the received signal for bit-1 is less  distinguishable than that for bit-0 with increasing $k_{-1}$, as shown in Fig.~\ref{fig:2}.

\begin{figure}[t!]
    \begin{center}
        \includegraphics[width=3.0 in,height=2.5in]{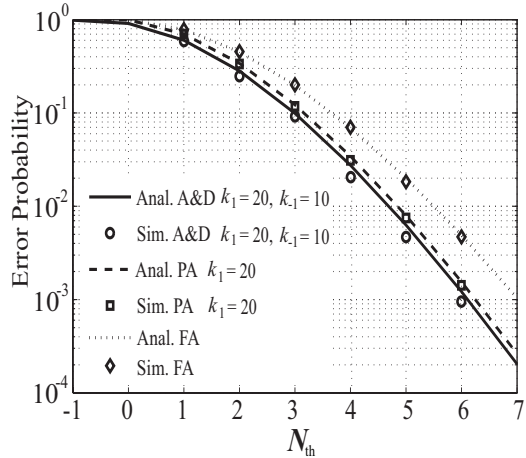}
        \caption{The error probability for the last transmit bit-0. }
        \label{fig:8}
    \end{center}
\end{figure}

\begin{figure}[t!]
    \begin{center}
        \includegraphics[width=3.0 in,height=2.5in]{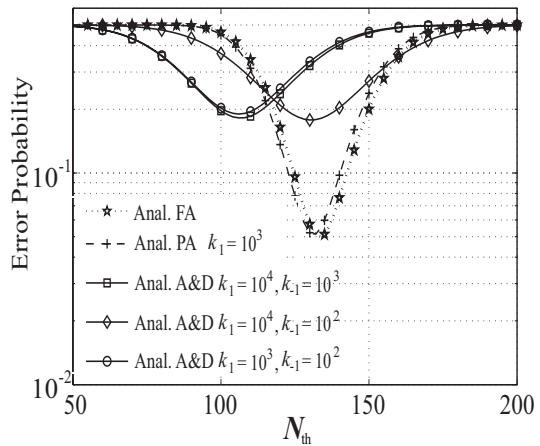}
        \caption{The error probability for the last random transmit bit. }
        \label{fig:9}
    \end{center}
\end{figure} 

 Fig.~\ref{fig:8} plots the error probability of the last transmit bit-0 for different types of  receivers with $N_b = 3$ bits transmitted. The error probability of the full adsorption receiver and the partial adsorption receiver are plotted using \eqref{error_bitFA1} and \eqref{error_bitFA0}.
  We see that the error probability  of the last transmit bit-0 decreases monotonically with increasing the  threshold $N_{\rm{th}}$.
  Interestingly, we see that the  error probability  of  the A$\&$D  receiver with  $k_1 =20$ $\mu \rm{m /s}$  and $k_{-1} =10$  $\rm{s^{-1}} $ outperforms that of the partial adsorption receiver with $k_1 =20$ $\mu \rm{m /s}$  and that of the full adsorption receiver, which is due to the higher tail effect from previous bits imposed on the partial adsorption receiver  and the full adsorption receiver compared to  that imposed on  the A$\&$D  receiver with   $k_{-1}=10$  $\rm{s^{-1}} $ as shown in  Fig.~\ref{fig:3}.
 
  
%

{Fig.~\ref{fig:9} plots the analytical results of the error probability of the last random transmit bit  for different types of  receivers with various $k_1$ and $k_{-1}$ using \eqref{overallerror}, considering that the analytical results have been verified in Figs.~\ref{fig:7} and ~\ref{fig:8}. We  set the parameters: $N_{\rm{tx}} = 1000$, $r_r = 5$ $\rm{\mu m}$, $r_0 = 10$ $\rm{\mu m} $, $D = 79.4$ $\rm{\mu m^2/s} $,   $T_s =0.002$ s,    $T_b =0.05$ s, $P_1 =P_0 =0.5$, and the first 2 bits ``1 1". Interestingly, we see that the error probability of the last random bit of the A$\&$D receiver is lower  than that of the full adsorption receiver and the partial adsorption receiver for $N_{\rm{th}}< 110$, and higher than that of the full adsorption receiver and the partial adsorption receiver for $N_{\rm{th}}> 120$. The observations are consistent with what we expected, and
 can be explained by  the advantage due to a lower tail at the  A$\&$D  receiver than that at the other receivers  at \emph{lower detection thresholds}, and  the advantage due to a higher peak at the other receivers than that at the  A$\&$D  receiver at \emph{higher detection thresholds.} We also observe that overall the PA and FA have lower optimal bit error probability. The A$\&$D receiver with $k_1=10^4$ and $k_{-1}=10^3$ achieves comparable   bit error probability value as that with $k_1=10^3$ and $k_{-1}=10^2$, which may due to having the same $k_1/{k_{-1}}$ ratio. }

\section{Conclusion}

In this paper, we modeled the diffusion-based MC system with the A$\&$D  receiver. We derived the exact expression for the   net number of adsorbed information molecules expected at the surface of the receiver. We also derived the asymptotic expression for the expected number of adsorbed information molecules as the bit interval goes to infinity.  We then derived the bit error probability of the A$\&$D  receiver. We also presented a simulation algorithm that captures the behavior of each information molecule with the stochastic reversible reaction at the receiver. 

Our results  showed that the 
error probability of the A$\&$D  receiver can be approximated by the Skellam distribution, and our derived analytical results closely matched our simulation results.   {  We revealed that the error probability of the A$\&$D  receiver for the transmit bit-1 improves with increasing  adsorption rate and with decreasing  desorption rate. More importantly,   the error probability of the A$\&$D receiver for the last transmitted bit is worse at higher detection thresholds but better at low detection thresholds than both the full adsorption and partial adsorption receivers. This is because the A$\&$D receiver observes a lower peak number of adsorbed molecules but then a faster decay.}
 Our analytical model and simulation framework provide a foundation for the accurate modeling and analysis of a more complex and realistic  receiver in molecular communication.
 
\appendices

\section{Proof of Theorem 1}\label{app_gsc_exact}

We first partition the spherically symmetric distribution into two parts using the method applied in \cite{schulten2000lectures}
\begin{align}
r{\cdot}C\left( {r,\left. t \right|{r_0}} \right) = r{\cdot}g\left( {r,\left. t \right|{r_0}} \right) 
+ r{\cdot}h\left( {r,\left. t \right|{r_0}} \right)
, \label{der1}
\end{align}
where
\begin{align}
g\left( {r,\left. {t \to 0} \right|{r_0}} \right) = \frac{1}{{4\pi {r_0}}}\delta \left( {r - {r_0}} \right)
, \label{der11}
\end{align}
\begin{align}
h\left( {r,\left. {t \to 0} \right|{r_0}} \right) = 0
. \label{der12}
\end{align}

Then, by substituting \eqref{der1} into \eqref{ficklaw},  we have
\begin{align}
\frac{{\partial \left( {r \cdot g\left( {r,\left. t \right|{r_0}} \right)} \right)}}{{\partial t}} = D\frac{{{\partial ^2}\left( {r \cdot g\left( {r,\left. t \right|{r_0}} \right)} \right)}}{{\partial {r^2}}}
, \label{der2}
\end{align}
and
\begin{align}
\frac{{\partial \left( {r \cdot h\left( {r,\left. t \right|{r_0}} \right)} \right)}}{{\partial t}} = D\frac{{{\partial ^2}\left( {r \cdot h\left( {r,\left. t \right|{r_0}} \right)} \right)}}{{\partial {r^2}}}
. \label{der3}
\end{align}

To derive $g\left( {r,\left. t \right|{r_0}} \right)$, we perform a Fourier transformation on $ r g\left( {r,\left. t \right|{r_0}} \right)$ to yield
\begin{align}
G\left( {k,\left. t \right|{r_0}} \right) = \int_{ - \infty }^\infty  {rg\left( {r,\left. t \right|{r_0}} \right){e^{ - ikr}}dr}
, \label{der41}
\end{align}
and
\begin{align}
r \cdot g\left( {r,\left. t \right|{r_0}} \right) = \frac{1}{{2\pi }}\int_{ - \infty }^\infty  {G\left( {k,\left. t \right|{r_0}} \right){e^{ikr}}dk} 
. \label{der42}
\end{align}

We then perform the Fourier transformation on \eqref{der2} to yield
\begin{align}
\frac{{dG\left( {k,\left. t \right|{r_0}} \right)}}{{dt}} =  - D{k^2}G\left( {k,\left. t \right|{r_0}} \right)
. \label{der5}
\end{align}

According to \eqref{der5} and the uniqueness of the Fourier transform, we derive 
\begin{align}
G\left( {k,\left. t \right|{r_0}} \right) = {K_g}\exp \left\{ { - D{k^2}t} \right\}
,\label{der51}
\end{align}
where $K_g$ is an undetermined constant.

The Fourier transformation  performed on \eqref{der11} yields
\begin{align}
G\left( {r,\left. {t \to 0} \right|{r_0}} \right) = \frac{1}{{4\pi {r_0}}}{e^{ - ik{r_0}}}
. \label{der111}
\end{align}

Combining \eqref{der51} and \eqref{der111}, we arrive at 
\begin{align}
G\left( {k,\left. t \right|{r_0}} \right) = \frac{1}{{4\pi {r_0}}}{e^{ - ik{r_0}}}\exp \left\{ { - D{k^2}t} \right\}
. \label{der6}
\end{align}

Substituting \eqref{der6} into \eqref{der42}, we find that
\begin{align}
r \cdot g\left( {r,\left. t \right|{r_0}} \right) = \frac{1}{{8\pi {r_0}\sqrt {\pi Dt} }}\exp \left\{ { - \frac{{{{\left( {r - {r_0}} \right)}^2}}}{{4Dt}}} \right\}
. \label{der52}
\end{align}

By performing the Laplace transform on  \eqref{der52}, we write
\begin{align}
\mathcal{L}\left\{ {r \cdot g\left( {r,\left. t \right|{r_0}} \right)} \right\} = \frac{1}{{8\pi {r_0}\sqrt {Ds} }}\exp \left\{ { - \left| {r - {r_0}} \right|\sqrt {\frac{s}{D}} } \right\}
. \label{der53}
\end{align}

We then focus on solving the solution $h\left( {k,\left. t \right|{r_0}} \right)$ by first performing the Laplace transform on $h\left( {k,\left. t \right|{r_0}} \right)$ and \eqref{der3} as
\begin{align}
H\left( {r,\left. s \right|{r_0}} \right)= \mathcal{L}\left\{ {  h\left( {r,\left. t \right|{r_0}} \right)} \right\} = \int_0^\infty  {h\left( {r,\left. t \right|{r_0}} \right){e^{ - s\tau }}d\tau } 
,\label{der7}
\end{align}
and 
\begin{align}
srH\left( {r,\left. s \right|{r_0}} \right) = D\frac{{{\partial ^2}\left( {rH\left( {r,\left. s \right|{r_0}} \right)} \right)}}{{\partial {r^2}}}
,\label{der8}
\end{align}
respectively.

According to \eqref{der8}, the Laplace transform of the solution with respect to the boundary condition in \eqref{der8} is 
\begin{align}
rH\left( {r,\left. s \right|{r_0}} \right) = f\left( s \right)\exp \left\{ { - \sqrt {\frac{s}{D}} r} \right\}
,\label{der9}
\end{align}
where $f\left( s \right)$ needs to satisfy the second initial condition in \eqref{initial2}, and the second boundary condition in  \eqref{boundary1} and \eqref{boundary2}.

Having the Laplace transform of $\{r \cdot g\left( {r,\left. t \right|{r_0}} \right)\}$ and ${  h\left( {r,\left. t \right|{r_0}} \right)}$ in \eqref{der53} and \eqref{der9}, and performing a Laplace transformation on \eqref{der1}, we derive 
\begin{align}
r\tilde C\left( {{r},\left. s \right|{r_0}} \right) = &G\left( {r,\left. s \right|{r_0}} \right) + rH\left( {r,\left. s \right|{r_0}} \right)
\nonumber\\ = &  \frac{1}{{8\pi {r_0}\sqrt {Ds} }}\exp \left\{ { - \left| {r - {r_0}} \right|\sqrt {\frac{s}{D}} } \right\} 
\nonumber\\&+ f\left( s \right)\exp \left\{ { - \sqrt {\frac{s}{D}} r} \right\},
\label{der10}
\end{align}
where $\tilde C\left( {r,\left. s \right|{r_0}} \right) = \int_0^\infty  {C\left( {r,\left. t \right|{r_0}} \right){e^{ - st}}dt} $.

{ We deviate from the method in \cite{schulten2000lectures},  and   perform the Laplace transform on the Robin boundary condition in \eqref{boundary3} to  solve $f\left( s \right)$, which yields}
\begin{align}
{{\tilde C}_a}\left( {\left. s \right|{r_0}} \right) = \frac{{{k_1}\tilde C\left( {{r_r},\left. s \right|{r_0}} \right)}}{{s + {k_{ - 1}}}}
,\label{der1111}
\end{align}
where  ${{\tilde C}_a}\left( {r,\left. s \right|{r_0}} \right) = \int_0^\infty  {{C_a}\left( {r,\left. t \right|{r_0}} \right){e^{ - st}}dt} $.

We then perform the Laplace transform on  the second initial condition in \eqref{initial2} and the second boundary condition in \eqref{boundary1} as
\begin{align}
{\left. {D\frac{{\partial \left( {\tilde C\left( {r,\left. s \right|{r_0}} \right)} \right)}}{{\partial r}}} \right|_{r = {r_r}}} = {k_1}\tilde C\left( {{r_r},\left. s \right|{r_0}} \right) - {k_{ - 1}}{{\tilde C}_a}\left( {\left. s \right|{r_0}} \right). \label{der121}
\end{align}

Substituting \eqref{der1111} into \eqref{der121}, we obtain
\begin{align}
{\left. {D\frac{{\partial \left( {\tilde C\left( {r,\left. s \right|{r_0}} \right)} \right)}}{{\partial r}}} \right|_{r = {r_r}}}   =   \frac{{{k_1}s}}{{s + {k_{ - 1}}}}\tilde C\left( {{r_r},\left. s \right|{r_0}} \right)
.\label{der12}
\end{align}

To facilitate the analysis, we express the Laplace transform on the second boundary condition  as 
\begin{align}
{\left. {\frac{{\partial \left( {r \cdot \tilde C\left( {r,\left. s \right|{r_0}} \right)} \right)}}{{\partial r}}} \right|_{r = {r_r}}} = \left( {1 + \frac{{{r_r}{k_1}s}}{{D\left( {s + {k_{ - 1}}} \right)}}} \right)\tilde C\left( {r,\left. s \right|{r_0}} \right)
.\label{der13}
\end{align}

Substituting \eqref{der10} into \eqref{der13}, we determine $f\left( s \right)$ as
\begin{align}
&f\left( s \right) = \frac{{\left( {\sqrt {\frac{s}{D}}  - \frac{1}{{{r_r}}} - \frac{{{k_1}s}}{{D\left( {s + {k_{-1}}} \right)}}} \right)}}{{\left( {\sqrt {\frac{s}{D}}  + \frac{1}{{{r_r}}} + \frac{{{k_1}s}}{{D\left( {s + {k_{-1}}} \right)}}} \right)}}\frac{{\exp \left\{ { - \left( {{r_0} - 2{r_r}} \right)\sqrt {\frac{s}{D}} } \right\}}}{{8\pi {r_0}\sqrt {Ds} }}.
\label{der15}
\end{align}

Having \eqref{der10} and \eqref{der15}, and performing the Laplace transform of the  concentration distribution,  we derive 
\begin{align}
&r \tilde C\left( {r,\left. s \right|{r_0}} \right) = \frac{1}{{8\pi {r_0}\sqrt {Ds} }}\exp \left\{ { - \left| {r - {r_0}} \right|\sqrt {\frac{s}{D}} } \right\}
\nonumber\\& \hspace{0.5cm} + \frac{1}{{8\pi {r_0}\sqrt {Ds} }}\exp \left\{ { - \left( {r + {r_0} - 2{r_r}} \right)\sqrt {\frac{s}{D}} } \right\}
\nonumber\\& \hspace{0.5cm} - \underbrace {  \frac{{2\left( {\frac{1}{{{r_r}}} + \frac{{{k_1}s}}{{D\left( {s + {k_{-1}}} \right)}}} \right)}{\exp \left\{ { - \left( {r + {r_0} - 2{r_r}} \right)\sqrt {\frac{s}{D}} } \right\}}}{{8\pi {r_0}\sqrt {Ds} }{\left( {\frac{1}{{{r_r}}} + \frac{{{k_1}s}}{{D\left( {s + {k_{-1}}} \right)}} + \sqrt {\frac{s}{D}} } \right)}}}_{Z\left( s \right)}.
\label{der16}
\end{align}

Applying the inverse Laplace transform  leads to 
\begin{align}
&r C\left( {r,\left. s \right|{r_0}} \right) = \frac{1}{{8\pi {r_0}\sqrt {\pi Dt} }}\exp \left\{ { - \frac{{{{\left( {r - {r_0}} \right)}^2}}}{{4Dt}}} \right\} +
\nonumber\\& \hspace{0.5cm} \frac{1}{{8\pi {r_0}\sqrt {\pi Dt} }}\exp \left\{ { - \frac{{{{\left( {r + {r_0} - 2{r_r}} \right)}^2}}}{{4Dt}}} \right\} -{\mathcal{L}^{ - 1}}\left\{ {Z\left( s \right)} \right\}.
\label{der17}
\end{align}

Due to the complexity of $Z(s)$, we can not derive the closed-form expression for its inverse Laplace transform ${f_z}\left( t \right) = {\mathcal{L}^{ - 1}}\left\{ {Z\left( s \right)} \right\}$. We  employ the Gil-Pelaez theorem \cite{wendel1961} for the characteristic function to derive  the cumulative distribution function (CDF) ${F_z}\left( t \right) $ as
\begin{align}
{F_z}\left( t \right) = & \frac{1}{2} - \frac{1}{\pi }\int_0^\infty  {\frac{{{\mathop{\rm Im}\nolimits} \left[ {{e^{ - jwt}}{\varphi _Z}\left( w \right)} \right]}}{w}} dw,
\nonumber\\   =  & \frac{1}{2} - \frac{1}{\pi }\int_0^\infty  {\frac{{{e^{ - jwt}}\varphi _Z^*\left( w \right) - {e^{jwt}}{\varphi _Z}\left( w \right)}}{{2jw}}} dw,
\label{der181}
\end{align}
where 
${\varphi _Z}\left( w \right)$ is given in \eqref{der19}.

Taking the derivative of ${F_z}\left( t \right)$, we derive the inverse Laplace transform of $Z(s)$ as
\begin{align}
&{f_z}\left( t \right) = \frac{1}{{2\pi }}\int_0^\infty  {\left( {{e^{ - jwt}}\varphi _Z^*\left( w \right) + {e^{jwt}}{\varphi _Z}\left( w \right)} \right)dw} .
\label{der18}
\end{align}

Combining \eqref{der17} and \eqref{der19}, we finally derive the expected time-varying spatial distribution in \eqref{der20}.
\bibliographystyle{IEEEtran}
\bibliography{IEEEabrv,ReceiverdoubleICCversion}

\balance
\end{document}